\documentclass[oneside, letterpaper]{amsart}
\usepackage[centertags]{amsmath}
\usepackage{amsfonts,amssymb}

\newcommand{\lprint}[1]{{\text{\bf\tt >[#1]<\ }}\label{#1}}

\renewcommand{\lprint}[1]{\label{#1}}

\theoremstyle{plain}                % title and number in bold, text italic
\newtheorem{thm}{Theorem}[section]
\newtheorem{lem}[thm]{Lemma}
\newtheorem{prop}[thm]{Proposition}
\newtheorem{cor}[thm]{Corollary}
\theoremstyle{definition}           % title and number in bold, text normal
\newtheorem{defn}[thm]{Definition}
\newtheorem{exam}[thm]{Example}
\newtheorem{prob}[thm]{Problem}
\newtheorem{sas}[thm]{Standing Assumption}
\theoremstyle{remark}               % title and number in italic, text normal
\newtheorem{rem}{Remark}

\numberwithin{equation}{section}

\newcommand{\R}{{\mathbb R}}
\newcommand{\N}{{\mathbb N}}
\newcommand{\PP}{{\mathbb P}}
\newcommand{\QQ}{{\mathbb Q}}
\newcommand{\EE}{{\mathbb E}}
\newcommand{\FF}{{\mathcal F}}
\newcommand{\GG}{{\mathcal G}}
\newcommand{\HH}{{\mathcal H}}
\newcommand{\DD}{{\mathcal D}}
\newcommand{\MM}{{\mathcal M}}
\newcommand{\TT}{{\mathcal T}}
\newcommand{\KK}{{\mathcal K}}
\newcommand{\KKB}{{\mathbb K}}
\newcommand{\FFF}{{\mathbb F}}

\renewcommand{\SS}{{\mathcal S}}
\renewcommand{\AA}{{\mathcal A}}

\newcommand{\scl}[2]{\langle #1,#2 \rangle} % scalar product
\newcommand{\sclb}[2]{\left\langle #1,#2 \right\rangle} % scalar product
\newcommand{\argmax}{\operatorname{argmax}}
\newcommand{\argmin}{\operatorname{argmin}}
\newcommand{\set}[1]{\left\{#1\right\}} % curly brackets
\newcommand{\sets}[2]{\set{#1\,:\,#2}} % a set with "such that"
\newcommand{\inds}[1]{ {\mathbf 1}_{\set{#1}}} % indicator of a curly set
\newcommand{\ind}[1]{ {\mathbf 1}_{{#1}}} % indicator of a set
\newcommand{\seq}[1]{\set{#1_n}_{n\in\N}}

\newcommand{\eps}{\varepsilon}
\newcommand{\ld}{\lambda}
\newcommand{\gm}{\gamma}

\newcommand{\esssup}{{\mathrm {esssup}}}
\newcommand{\essinf}{{\mathrm {essinf}}}

\newcommand{\dd}[1]{{  d}#1} % differential
\newcommand{\rn}[2]{\frac{\dd{#1}}{\dd{#2}}} % derivative
\newcommand{\el}{{\mathbb L}} %l-pees
\newcommand{\lzer}{\el^0}
\newcommand{\lone}{\el^1}

\newcommand{\linf}{\el^{\infty}}
\newcommand{\linfd}{(\linf)^*} % l-infinity dual
\newcommand{\pd}{\partial_2}
\newcommand{\ea}[1]{\begin{eqnarray*} #1 \end{eqnarray*}}
\newcommand{\sint}[1]{[\!\![#1]\!\!]} % Stochastic interval
\newcommand{\sinto}[1]{[\!\![#1)\!\!)} % Semiopen Stochastic interval

\newcommand{\EN}{{\mathcal E}} % endowment process
\newcommand{\XX}{{\mathcal X}} % set of admissible wealths
\newcommand{\YY}{{\mathcal Y}} % dual supermartingales
\newcommand{\yq}{ Y^{\QQ}} % radon-nykodim of Q
\newcommand{\prf}[1]{ ( #1 )_{t\in [0,T]}} % process form - puts parenthesis
\newcommand{\qrpt}{\rn{(\QQ|_{\FF_t})^r}{(\PP|_{\FF_t})}} % definition of yqt
\newcommand{\ym}{\YY^{\MM}} % densities of mart. measures
\newcommand{\yd}{\YY^{\DD}} % densities of mert. charges
\newcommand{\yh}{\hat{Y}^{\QQ}}
\newcommand{\ynn}{Y^{(n)}}

\newcommand{\qfs}{\QQ|_{\FF_s}}
\newcommand{\qft}{\QQ|_{\FF_t}}

\newcommand{\ei}[2]{\EE{\int_0^{#1} #2}}
\newcommand{\uu}{\underline{U}}
\newcommand{\ou}{\overline{U}}
\newcommand{\uv}{\underline{V}}
\newcommand{\ov}{\overline{V}}
\newcommand{\eit}[1]{\EE\left[\int_0^T {#1}\,\mu(\dd{t})\right]}

\newcommand{\hq}{\hat{\QQ}}
\newcommand{\hqy}{\hq^{y}}
\newcommand{\qn}{{{\QQ}^{(n)}}}
\newcommand{\qnn}{\{\qn\}_{n\in\N}}

\newcommand{\yqn}{Y^{\qn}}
\newcommand{\yqy}{Y^{\hqy}}
\newcommand{\qe}{\QQ^{\eps}}
\newcommand{\yqe}{Y^{\qe}}
\newcommand{\norm}[1]{{||#1||}}

\newcommand{\oi}{\overline{I}}
\newcommand{\hc}{\hat{c}}

\newcommand{\xpi}{X^{x,H}}

\newcommand{\hH}{\hat{H}}

\newcommand{\FFt}{\FF_t}

\newcommand{\RR}{{\mathcal R}}

\newcommand{\uc}[1]{\eit{U(t,#1(t))}}
\newcommand{\vq}[2]{\eit{V(t,#1 Y^{#2}_t)}}

\begin{document}

% [Runnung head] + Title

\title[Random Endowment]{Optimal consumption from investment and random endowment in
incomplete semimartingale markets}

% Information for first author

\author{Ioannis Karatzas}
\address{Ioannis Karatzas\\ Departments of Mathematics and Statistics\\
Columbia University\\ New York\\ NY 10027}
\email{ik@math.columbia.edu}

% Information for second author

\author{Gordan \v Zitkovi\' c}
\address{Gordan \v Zitkovi\' c\\ Department of Statistics\\
Columbia University\\ New York\\ NY 10027}
\email{gordanz@stat.columbia.edu}

% General info
\date{\today}

% Abstract

\begin{abstract} We consider the problem of maximizing expected utility
from consumption in a constrained incomplete semimartingale market
with a random endowment process, and establish a general existence
and uniqueness result using techniques from convex duality.  The
notion of asymptotic elasticity of Kramkov and Schachermayer is
extended to the time-dependent case. By imposing no smoothness
requirements on the utility function in the temporal argument, we
can treat both pure consumption and combined consumption/terminal
wealth problems, in a common framework. To make the duality
approach possible, we provide a detailed characterization of the
enlarged dual domain which is reminiscent of the enlargement of
$\lone$ to its topological bidual $\linfd$, a space of
finitely-additive measures. As an application, we treat the case
of a constrained It\^ o-process market-model.
\end{abstract}

% Creates the Title-page

\maketitle

%---------------------------------------------------------
% PART III - Section 1 - Introduction
%---------------------------------------------------------

\section{Introduction}\lprint{Intro}

Both modern and classical  theories of economic behavior use
utility functions to describe the amount of  ``satisfaction'' of
financial agents depending on their wealth or consumption rate.
 Starting with an initial endowment,
an agent is faced with the problem of distributing
 wealth among financial assets with different degrees of uncertainty.
 If the
market is arbitrage-free, the agent can never ``beat the market'',
but may still invest in such a way as to maximize expected
utility. A considerable body of literature
 has been devoted to this subject. First to
consider the utility maximization problem in  continuous-time
stochastic financial market models was  Merton in \cite{Mer69},
\cite{Mer71}. He used a strong assumption (usually not justified
in practice) that stock-prices are governed by Markovian dynamics
with constant co\" efficients. In this way he could use the
methods of stochastic programming and in particular, the
Bellman-Hamilton-Jacobi equation of dynamic programming. More
recently, a ``martingale'' approach to the problem in complete
It\^ o-process markets was introduced by Pliska \cite{Pli86},
Karatzas, Lehoczky and Shreve \cite{KarLehShr87} and Cox and Huang
\cite{CoxHua89}, \cite{CoxHua91}. They related the marginal
utility  from the terminal wealth of the optimal portfolio to the
density of the (unique) martingale measure, using powerful
convex-duality techniques. Difficulties with this approach arise
in incomplete markets. The main idea here is to use the convex
nature of the problem, to formulate and solve a dual variational
problem, and then proceed as in the complete case. In
discrete-time and on a finite probability space, the problem was
studied by He and Pearson \cite{HePea91}, and   in a
continuous-time model by  of G.-L. Xu in his doctoral dissertation
\cite{Xu90},  by He and Pearson \cite{HePea91a} and by Karatzas,
Lehoczky, Shreve and Xu \cite{KarLehShrXu91}. In the paper
\cite{KraSch99}, Kramkov and Schachermayer solve the problem in
the context of a general incomplete semimartingale financial
market. They show that a necessary and sufficient condition for
the existence of an optimal solution is {\em reasonable asymptotic
elasticity} of the utility function. This is an analytic condition
on the behavior of the utility function at infinity, which
excludes certain pathological situations. These authors also show
that the set of densities of local martingale measures is too
small to host the solutions of the dual problem. Thus, they
enlarge it to a suitably chosen set $\YY$ of supermartingales, in
a manner  reminiscent of enlarging $L^1$ to its topological bidual
$(L^{\infty})^*$. Although these supermartingales cannot be used
directly as pricing rules for derivative securities, Kramkov and
Schachermayer show this is possible under an appropriate change of
num\' eraire.

When, in addition to initial wealth, the agent faces an uncertain
random intertemporal endowment, the situation becomes technically
much more demanding and the gap between complete and incomplete
markets even more apparent. In the complete market setting, the
entire uncertain endowment can be ``hedged away'' in the market,
and the problem becomes equivalent to the one where the entire
endowment process is replaced by its present value, in the form of
an augmented initial wealth. A self-contained treatment of this
situation, in It\^ o-process models for financial markets can
be found in Section 4.4 of the monograph by Karatzas and Shreve
\cite{KarShr98}. An otherwise complete market with random
endowment, where the incompleteness is introduced through
prohibition of borrowing against future income, is dealt with in
\cite{ElkJea98}. In incomplete markets, several authors consider
this problem in various degrees of generality. We mention Cuoco
who deals with a cone-constrained It\^ o-process market with
random endowment in \cite{Cuo97} - he attacks directly the primal
problem circumventing the duality approach altogether, at the cost
of rather strict restrictions on the utility function. A
definitive solution to the problem of maximizing of utility from
terminal wealth in incomplete (though not constrained in a more
general way) semimartingale markets with random endowment is
offered in \cite{CviSchWan99}. The main contribution of that paper
is the introduction of finitely-additive measures into the realm
of optimal stochastic control problems encountered in mathematical
finance. The essential difference between
utility maximization with and without random endowment is probably
best described by the authors of \cite{CviSchWan99}:
\begin{quote}{\em `` it was not important in the analysis of
\cite{KraSch99} where the `singular mass of $\hat{\QQ}$ has
disappeared to'. In the present paper this becomes very important
\ldots [it] acts on the accumulated random endowment and can be
located in $(\linf)^*$''.}\end{quote} We finally mention
\cite{Sch00} as an extensive survey of the optimal investment
theory.

This paper strives to complement the existing results
 in several ways.
First, we incorporate inter-temporal consumption in the
optimization problem. We are dealing with an agent investing in an
incomplete market, where prices are modelled by an arbitrary
semimartingale with right-continuous and left-limited paths. From
the present moment to some finite time horizon $T$, our agent is
not only deciding how to manage a portfolio by dynamically
readjusting the positions in various financial assets, but also
choosing a portion of wealth to be consumed and not further
reinvested. The agent also has to take into account the
uncertainty in the random endowment stream. It is from this
consumption, or from consumption and terminal wealth, that
utility is derived. We allow the utility function to be random,
reflecting the changes in agent's risk-preferences from one time
to another. In a departure from existing theory, we do not impose
any smoothness on the  utility function in its temporal argument.
As a result, we have a common framework for problems that involve
consumption only {\em and}  for problems that involve both
consumption and terminal wealth. In addition to dealing with an
inherently incomplete semimartingale market-model, we impose
convex cone constraints on the investment choices the agent is
facing. In this way we can model incompleteness and prohibition of
short-sales, to name only two.

For utility functions we formulate the concept of asymptotic elasticity
and, under an appropriate condition of ``reasonable asymptotic
 elasticity'', we establish existence and
uniqueness of the optimal consumption-investment strategy. In
\cite{KraSch99} it
 was only the terminal value of a dual process that appeared in the analysis,
 the dual domain
$\set{ Y_T\,:\, Y\in\YY}\subseteq L^0_+$ being endowed with the
topology of convergence in probability. The more difficult
situation in \cite{CviSchWan99} required the dual domain to be
extended to the closure of the set of all equivalent martingale
measures in $(\linf)^*$ - a space whose elements are {\em
finitely-additive set-functions}. Abusing terminology slightly, we
shall call such set-functions ``finitelly-additive  measures''. In
our case, we have to mimic the natural correspondence between
measures and uniformly integrable martingales in the
finitely-additive world. It turns out that the right choice
consists of a dual domain, inhabited by finitely-additive
measures, and coupled with supermartingales corresponding to the
Radon-Nikodym derivatives of their regular parts. We prove
rigorously that these supermartingales essentially correspond to
the supermartingales in the set $\YY$ defined in \cite{KraSch99}.
The basic tool in this endeavor is the Filtered Bipolar Theorem of
\cite{Zit00}.

As applications of our results, we treat two special cases - a
constrained It\^ o-process market, where we prove that the optimal
dual process is always a local martingale, and the ``totally
incomplete'' case of Lakner and Slud (\cite{LakSlu91}), where the
agent is not allowed to invest in the stock-market at all.

We should stress that one   main motivation behind this work is
the r\^ ole it plays as a necessary step for an offensive on the
problem of existence and uniqueness for equilibrium in
continuous-time incomplete markets with random endowments, a task
we plan to attempt in future research.

The part of our analysis dealing with duality, and especially the structure of the proof of the main result, is closely based on and inspired by the expositions in \cite{KraSch99} and
\cite{CviSchWan99}. In Section \ref{FinMarCh} we set up the
market-model, and present a characterization of admissible
consumption strategies. Section \ref{OptProb} displays our main result and Appendix \ref{ProofA} its proof. In Section
\ref{Examples} we give an application of our results through two examples.

%----------------------------------------------------------------
% PART IV - Section 2 - The Model
%----------------------------------------------------------------

\section{The model}\lprint{FinMarCh}

\subsection{The financial market}\lprint{FinMar1}
We introduce a model for a financial market consisting of
\begin{itemize}
\item[(i)] a positive, adapted process $B=\prf{B_t}$ with paths that are RCLL
(Right-Continuous on $[0,T)$, with Left-Limits everywhere on
$(0,T]$) and  uniformly bounded from above and away from zero. We
interpret $B$ as the num\' eraire asset - a bond, for example.
\item[(ii)] a RCLL-semimartingale $S=\prf{S_t}$ taking values in $\R^d$; its
component processes represent the prices  of $d$ risky assets,
discounted in terms of the num\' eraire $B$.
\end{itemize}
All processes are defined on a stochastic base
$(\Omega,\FF,\prf{\FF_t},\PP)$ with a finite {\bf time horizon}
$T>0$, and the filtration $\FFF\triangleq\prf{\FF_t}$ satisfies
the usual conditions; $\FF_0$ is the completion of the trivial
$\sigma$-algebra.

We concentrate our attention on a financial agent endowed with
initial wealth $x>0$ and a random {\bf cumulative endowment process}
$\EN=\prf{\EN_t}$ - in that the total (cumulative) amount of
endowment received by time $t$ is $\EN_t$. We assume that
$\EN_0=0$ and $\EN$ is nondecreasing, $\FFF$-adapted, RCLL  and
uniformly bounded from above, i.e., $\EN_T\in\linf_+(\PP)$.
Similarly to the price-process $S$, we assume that $\EN$ is
already discounted (denominated in terms of $B$).

Faced with inherent uncertainty in  future endowment, the agent
dynamically adjusts positions in different financial assets and
designates a part of wealth for immediate consumption, in the
following manner:
\begin{itemize}
%\noindent
\item[(a)] the agent chooses an $S$-integrable and $\FFF$-predictable process $H$
taking values in $\R^d$. The process
$H$ has a natural interpretation as {\bf portfolio process}; in
other words, the $i^{\text{th}}$ component of $H_t$ is the number
of shares of stock $i$ held at time $t$.

To exclude pathologies such as doubling schemes, we choose to
impose the condition of {\bf admissibility} on the agent's choice
of portfolio process $H$, by requiring that the {\bf gains process}
$\int_0^{\cdot} H_u\,\dd{S_u}$ be uniformly bounded from below by
some constant (for the theory of stochastic integration
with respect to RCLL semimartingales, and the related notion of
integrability, the reader may consult \cite{Pro90}).
Moreover, we ask our agent to obey the investment
restrictions imposed on the structure of the market, by choosing
the portfolio process $H$ in a closed convex cone $\KK\subseteq
\R^d$. The set $\KK$ represents constraints on  portfolio choice,
and can be used to model, for example, short-sale constraints or
unavailability of some stocks for investment.
%\noindent
\item[(b)] apart from the choice of   portfolio process, the agent
chooses a nonnegative, nondecreasing $\FFF$-adapted RCLL  process
$C=\prf{C_t}$. The {\bf cumulative consumption process} $C$ represents the
total amount (just like $S$ and $\EN$, already discounted by
$B$) spent on consumption, up to and including time $t$.
\end{itemize}

A pair $(H,C)$ that satisfies (a) and (b) above, is called an {\bf
investment-con\-sum\-pti\-on strategy}. The wealth of an agent
that employs the investment-con\-sum\-pti\-on strategy $(H,C)$ is
given by \begin{equation} \lprint{defadm}  W_t^{H,C}\triangleq
x+\EN_t+\int_0^{t} H_u\,\dd{S_u}-C_t,\ \ 0\leq t\leq
T.\end{equation} If the strategy $(H,C)$ is such that the
corresponding wealth process $W^{H,C}$ satisfies $W^{H,C}_T\geq 0$
a.s., we say that $(H,C)$ is an {\bf admissible strategy}. If, for
a consumption process $C$, we can find a portfolio process $H$
such that $(H,C)$ is admissible, we call $C$  an {\bf admissible
consumption process}, and say that $C$ can be {\bf
financed} by $x+\EN$ and $H$. Let $\mu$ be an {\bf admissible measure},
i.e., a probability measure on $[0,T]$, diffuse on $[0,T)$, such
that $\mu([0,t])<1$ for all $t<T$. For such a measure we define
the {\bf support} $\mathrm{supp}\,\mu$ to be $[0,T]$ if $\mu$
charges $\set{T}$, and $[0,T)$ otherwise.

We shall be mostly interested in admissible consumption
processes $C$ that can be expressed as
\begin{equation*} C_t=\int_0^t c(u)\, \mu({\mathrm du}),\ \text{\
$0\leq t\leq T$.}\end{equation*} The set of all densities $c(\cdot)$ of
such processes will be denoted by $\AA^{\mu}(x+\EN)$. We allow for
bulk consumption at the terminal time in order to be able to deal
later on with utility from the terminal wealth and/or from
consumption, in the same framework.

% Remark 1:
\begin{rem}{\rm  Even though we allow debt to incur before time
$T$, the agent must invest in such a way as to be able to post a
non-negative wealth  by the end of the trading horizon, with
certainty. Furthermore, the boundedness of the process
$\EN=\prf{\EN_t}$ guarantees that the negative part of the wealth
will remain bounded by a constant (a weak form of ``constrained
borrowing'').

The following notation will be used repeatedly in the sequel:
\begin{eqnarray}\lprint{XX}\nonumber
\XX\triangleq&&\!\!\!\!\!\!\!\!\!{\Big\{}x+\int_0^{\cdot}H_u\,\dd{S_u}\,:\,
\text{$H$ is predictable and $S$-integrable, $H_t\in\KK$ a.s.
}\\\hskip2cm && \text{
 ~~~for every $t\in [0,T]$, $x\geq 0$, and}\
x+\int_0^{\cdot}H_u\,\dd{S_u}\ \text{is nonnegative}{\Big
\},}\end{eqnarray}}\end{rem}
\subsection{The optimization problem}

Let us introduce now a preliminary version of the optimization
problem, and lay out an outline of its solution. The goal is to
find a consumption-density process $\hc^x(\cdot)$, financed by the
initial wealth $x$ and the random endowment $\EN$, which maximizes
the expected utility from consumption - the average felicity of an
agent who follows the consumption strategy $\hc^x(\cdot)$. The expected
utility from a consumption density process $c(\cdot)$ is given by
\[ \eit{U(t,c(t))},\] where $U$ denotes a (random) utility function and $\mu$
a utility measure. We postpone discussion of the definition and
regularity properties of $U$ until Section \ref{OptProb}. In this
notation, \begin{equation}
\hc^x=\argmax_{c\in\AA^{\mu}(x+\EN)}\eit{U(t,c(t))}.\lprint{op1}\end{equation}
As it is customary in the duality approach to stochastic
optimization, we introduce a problem dual to (\ref{op1}) by
setting
\[ Y^{\hq^y}=\argmin_{\QQ\in\DD}\left[
\EE\int_0^T{V(t,Y^{\QQ}_t)}\,dt+y\scl{\QQ}{\EN_T}\right].\] Here
$\DD$ denotes the  domain for the dual problem; it is the closure
of the set of all supermartingale measures for the stock process
$S$. The process $Y^{\QQ}$ is a supermartingale version of the
density process of $\QQ$, and $V$ is the convex conjugate of $U$.

In the following subsections, we introduce and describe the dual
domain $\DD$ in detail, and establish some of its properties - the
prominent one being weak * compactness. It is precisely this
compactness property that will ensure the existence of a solution
to the dual problem and  - through standard tools of convex
duality - the existence of an optimal consumption process $\hc^x$
for any positive initial wealth $x$.

\subsection{Connections with Stochastic Control Theory}

The portfolio process $H$ serves as the analogue of the
control-process in Stochastic Control Theory. It is important, though, to stress that we are not dealing here with a {\em partially (incompletely) observed} problem (a terminology borrowed again from Control Theory). Incomplete
markets in Mathematical Finance correspond to a setting, in which the controller has full information about many aspects of the system (the market), but various exogenously imposed constraints (taxation, transaction costs, bad credit rating, legislature, etc.) prevent  him/her from choosing the control (portfolio) outside a given constraint set. In fact, even without government-imposed portfolio constraints, financial markets will typically not offer tradeable assets corresponding to a variety of sources of uncertainty (weather conditions, non-listed companies, etc.) The financial agent will still observe many of these sources, as their uncertainty evolves, but will typically not be able to ``trade in all of them'', as it were.

This fundamental nature of financial markets is reflected in our
modelling: in Sections 1, 2 and 3, we allow the filtration
$\mathbb{F}$ (with respect to which the controls are adapted) to be possibly larger that the filtration generated by the
stock-price process $S$. The only requirement we impose, in the next subsection, is the one of \textsl{absence of arbitrage}, the fulfilment of which depends heavily on the choice of filtration $\mathbb{F}$. To sum up, the {\em observables} in financial modelling constitute a much larger class than the mere stocks we are allowed to invest in. With such an understanding, our portfolios {\em are} adapted only to the observables of the system. Such a setting corresponds to the well-established control-theoretic notion of admitting ``open loop'' controls in our analysis.

\smallskip
In the more specialized setup of Section 4, the filtration
$\mathbb{F}$ is taken as the augmentation of the filtration
generated by the Brownian motions driving the stock-prices,
assuming as we do in the beginning of Subsection 4.1  that the
volatility matrix process $\sigma(t)$ is {\em non-singular} a.s., for each $t$. At the level of generality considered in the paper, the filtration corresponding to the stock prices will be smaller than the filtration generated by the Brownian motion. But the two filtrations \textsl{are} actually the same, when interest-rates, volatilities and appreciation-rates are functions of past-and-present stock prices; this includes the case of Markovian or deterministic coefficients. In this case, ``open loop'' and ``closed loop'' (i.e., $S$-adapted) controls, actually coincide.

\medskip

Finally, we would like to stress that market incompleteness is the main source of technical and conceptual problems we had to overcome in this work, whereas the case of complete markets has been well studied by many authors before; see, for instance, Chapters 3 and 4 in [KS98]. All of our results concerning the structure of the dual domain (as well as the introduction of the dual domain in the first place) are consequences of the incompleteness of the market. We are actually allowing for \textsl{two} separate sources of incompleteness - the general structure of the stock-prices, as well as the portfolio constraints in the form of the cone $\mathcal{K}$. By choosing $\,\mathcal{K} = \mathbb{R}^n \times \{0\} \times \cdots \times \{0\}$ for some $n=1, \cdots d-1$, we capture exactly the setting of an incomplete market with $n$ stocks, and with $d>n$ sources of
randomness that affect the coefficients in the model.

\subsection{Absence of arbitrage, finitely-additive set-fun\-cti\-ons, and
the dual domain}

In order to make possible a meaningful mathematical treatment of
the optimization problem, we outlaw arbitrage opportunities by
postulating the existence of an {\bf equivalent supermartingale
measure}, i.e., a probability measure on $(\Omega,\FF)$, equivalent
to $\PP$, under which the elements of the set $\XX$ in (2.2) become
supermartingales. The set of all equivalent supermartingale
probability measures will be denoted by $\MM$, and we shall assume
throughout that $\MM\neq\emptyset$. A detailed treatment of the
connections between various notions of arbitrage and the existence
of equivalent martingale (local martingale, supermartingale)
measures, culminating with the Fundamental Theorem of Asset
Pricing, can be found in \cite{DelSch93} and \cite{DelSch98}.

As  was pointed out in \cite{CviSchWan99}, the duality treatment
of   utility maximization requires a nontrivial enlargement of
$\MM$: this space turns out to be too small, in terms of closedness
and compactness properties. Accordingly, we define $\DD$ to be the
$\sigma(\linfd,\linf)$-closure of $\MM$ in $\linfd$ -- the
topological dual of $\linf$ -- where $\MM$ is canonically
identified with its embedding into $\linfd$. We shall denote by
$\linfd_+$ the set of non-negative elements in $\linfd$. In the following
proposition we collect some properties of $\linfd$, $\linfd_+$\,, and
$\DD$; more information about $\linfd$ can be found in \cite{RaoRao83}.

\begin{prop}\lprint{charge}
\begin{enumerate}
\item[(i)] The space $\linfd$ consists of finite, finitely-additive measures
 on $\FF$, which assign the value zero to $\PP$-null subsets of $\FF$.
\item[(ii)] Under the canonical pairing $ \scl{\ }{\ }\,:\, \linfd\times\linf
\to\R$,  the relation $\scl{\QQ}{1}=1$ holds for all $\QQ\in\DD$.
In other words, with the notation $\, \QQ(A) \triangleq
\scl{\QQ}{1_A}\,$ for $\,A \in \FF $ and $\QQ \in \linfd$, we
have $\,\QQ(\Omega )=1$ for all \ $\QQ\in\DD$.
\lprint{probcharge}
\item[(iii)] $\DD$ is weak * (i.e., $\sigma(\linfd,\linf)$)\,--\,compact.
\item[(iv)] Every element $\QQ$ of $\linfd_+$ admits a unique decomposition
\[
\QQ=\QQ^r+\QQ^s, \ \  \ \text{with}\ \ \ \QQ^r,\QQ^s\in\linfd_+\,,
\]
where the {\bf regular part} $\QQ^r$ is the maximal  {\rm
countably-additive} measure on $\FF$ dominated by $\QQ$, and the
{\bf singular part} $\QQ^s$ is {\rm  purely finitely-additive},
i.e., does not dominate any nontrivial countably-additive measure.
\item[(v)] \lprint{wang} $\QQ\in\linfd_+$  is singular, if
and only if for any $\eps>0$ there exists $A_{\eps}\in\FF$ such
that $\PP(A_{\eps})>1-\eps$ and $\QQ(A_{\eps})=0$.
\item[(vi)] \lprint{wang2} Suppose a bounded sequence $\seq{\QQ}$ in $\linfd_+$
is such that $\frac{\dd{\QQ_n^r}}{\dd{\PP}}\to f$ a.s., for some
$f\geq 0$. Then any weak * cluster point $\QQ$ \ of \ $\seq{\QQ}$
satisfies \ $\frac{\dd{\QQ^r}}{\dd{\PP}}= f$ a.s. where $\QQ^r$
denotes the regular part of $\QQ$.
\item[(vii)] The regular-part operator $\,\QQ\mapsto\QQ^r$ is additive on
$\linfd_+$.
\end{enumerate}
\end{prop}
\begin{proof}\
\begin{enumerate}
\item[(i)] See \cite{RaoRao83}, Corollary 4.7.11.
\item[(ii)] Follows from density of $\MM$ in $\DD$.
\item[(iii)] This is the content of Alaoglu's theorem (see \cite{Woj96}, Theorem 2.A.9).
\item[(iv)] See Theorem 10.2.1 in \cite{RaoRao83}.
\item[(v)] See Lemma A.1. in \cite{CviSchWan99}.
\item[(vi)] See Proposition A.1. in \cite{CviSchWan99}.
\item[(vii)] Let $\QQ$ and $\RR$ be elements of $\linfd_+$.
It is clear that $\QQ^r+\RR^r$ is a countably additive measure
dominated by $\QQ+\RR$, so $(\QQ+\RR)^r\geq \QQ^r+\RR^r$. For the
equality, it is enough to show that
$(\QQ+\RR)-(\QQ^r+\RR^r)=\QQ^s+\RR^s$ is singular. For any
$\eps>0$, by (v), we can find sets $A_{\eps}$ and $B_{\eps}$ such
that $P(A_{\eps})>1-\frac{\eps}{2}$,
$P(B_{\eps})>1-\frac{\eps}{2}$ and
$\QQ^s(A_{\eps})=\RR^s(B_{\eps})=0$. With $C_{\eps}\triangleq
A_{\eps}\cap B_{\eps}$ we have $P(C_{\eps})>1-\eps$ and
$(\QQ^s+\RR^s)(C_{\eps})=0$\,; this completes the proof, by appeal to
(v). \end{enumerate}  \end{proof}

\begin{rem}{\rm
In the light of property (ii) we may interpret the elements of
$\DD$ as finitely-additive probability measures on $\FF$, weakly
absolutely continuous with respect to $\PP$. }\end{rem}

For our analysis, it will be necessary to associate a nonnegative RCLL
supermartingale $\yq=\prf{\yq_t}$ to every $\QQ\in\DD$. For $\QQ\in\MM$,
this process is just the RCLL-modification of the martingale
$\prf{\EE[\frac{d\QQ}{d\PP}|\FFt]}$. For general $\QQ\in\linfd_+$,
the construction of $\yq$ is rather delicate (cf. (\ref{yqdef})
below). To make headway on this issue, we let $\QQ^r$ denote the
regular part of $\QQ$ and, for any $\sigma$-algebra
$\GG\subseteq\FF$, we denote by $\QQ|_{\GG}$ the restriction of
the set-function $\QQ$ to $\GG$. Since the regular-part operator
$\QQ\mapsto\QQ^r$ depends nontrivially on the domain of $\QQ$, we
stress  that $(\QQ|_{\GG})^r$ stands for a countably-additive
measure on $\GG$ and, in general, does {\em not} equal
$\QQ^r|_{\GG}$\,: the regular-part and restriction operations do not
commute, in general. In fact, we have the following result:

\medskip
\begin{prop} \lprint{monot} For any two $\sigma$-algebras
$\GG\subseteq\HH$ and every $\QQ\in\linfd$, we have
$(\QQ|_{\GG})^r\geq (\QQ|_{\HH})^r|_{\GG}.$
\end{prop}
\begin{proof} By definition, $(\QQ|_{\GG})^r$ is the maximal
countably-additive measure on $\GG$ dominated by $\QQ$, so it must
dominate $(\QQ|_{\HH})^r|_{\GG}$ -- another countably-additive
measure on $\GG$ dominated by $\QQ$.
\end{proof}

 For $\QQ\in\DD$ we define the process
 \begin{equation} \lprint{ldef} L^{\QQ}_t\triangleq \qrpt,\ \ t\in [0,T].\end{equation}
 It is exactly the property from Proposition
 \ref{monot} that makes then the process defined by
\begin{equation}
\yq_t \, \triangleq \liminf_{q\searrow t,\  \text{$q$ is rational}}
L^{\QQ}_q,\ \  0\leq t< T,\ \ \ ~~ \text{and} ~~ \ \ \ \yq_T\triangleq
L^{\QQ}_T\lprint{yqdef}
\end{equation} a RCLL  supermartingale. This, seemingly unnatural,
 regularization through the limit-inferior in
(\ref{yqdef}) is necessary, since there is no guarantee that an
RCLL-modification exists for  the process $L^{\QQ}$. Appendix I,
theorem 4, p. 395 and Theorem 10, p. 402 in \cite{DelMey82}
establish good measurability properties of the processes involved,
as well as the fact that the limit-inferior in (\ref{yqdef}) is
actually a true limit for every $t\in [0,T)$, on a subset of
$\Omega$ of full probability. When $\QQ\in\MM$, it is immediate
that the process $\yq=\prf{\yq_t}$ of (\ref{yqdef}) is the
RCLL-modification of the martingale
$\prf{\EE[\rn{\QQ}{\PP}|\FF_t]}$.

We define the two sets of processes
\begin{equation}\lprint{yqdefn} \ym\triangleq\sets{\yq }{\QQ\in\MM}\ \
\text{\ \ and\ \ }\ \ \yd\triangleq\sets{\yq
}{\QQ\in\DD} \supseteqq \ym\,.\end{equation} The following proposition goes deeper
into the properties of the elements of $\yd$. It shows that the
regularization in the definition (\ref{yqdef}) of the process
$\yq$ is, in fact, a harmless operation.

\smallskip
% Proposition 2.3
\begin{prop}\lprint{pryq}
\begin{itemize}
\item[(a)] For $\QQ\in\DD$, there exists a
countable set $K\subseteq [0,T)$, such that $\, \yq_t=L^{\QQ}_t$ for all $t\in
[0,T]\setminus K$,  almost surely. In particular,
$\yq=L^{\QQ}$ $(\mu\otimes\PP)$-a.e., for any admissible measure
$\mu$.
\item[(b)] For every stopping time $S$, we have\
$\, \yq_S\leq \rn{(\QQ|_{\FF_S})^r}{(\PP|_{\FF_S})}$ a.s.
\end{itemize}
\end{prop}

\begin{proof}
\begin{itemize}
\item[(a)]
Let $K$ be the set of discontinuity  points  of the decreasing
function $ t\mapsto \EE[L^{\QQ}_t]= (\QQ|_{\FF_t})^r(\Omega),$ on
$[0,T)$; this set is at most countable. For every $t<T$, Fatou's
lemma gives
\begin{equation}\lprint{fh} \yq_t\leq\liminf_{q\searrow t,\
\text{q is rational}} \EE[L^{\QQ}_q|\FF_t]\leq
L^{\QQ}_t.\end{equation} On the other hand, for any sequence of
rationals $\seq{q}$ with $q_n\searrow t$,
$\set{L^{\QQ}_{q_n}}_{n\in\N}$ is a backward supermartingale
bounded in $\lone$, so that $L^{\QQ}_{q_n}\to \yq_t$ both in
$\lone$ and a.s., thanks to  the Backward Supermartingale
Convergence Theorem (see \cite{Chu74}, Theorem 9.4.7, page 338).
For each $t\in [0,T]\setminus K$ we have thus $\EE[Y^{\QQ}_t]=\EE[
L^{\QQ}_t]$ which, together with (\ref{fh}) and the fact that $K$ is
at most countable, completes the proof of (a).
\item[(b)]
For an arbitrary stopping time $S$, and $n\in\N$, we put
$S^n=(2^{-n}\lfloor 2^n S+1\rfloor)\wedge T$, so that $S\leq
S^n\leq S+2^{-n}.$ Therefore, $\set{S^n}_{n\in\N}$ is a sequence
of stopping times with finite range, a.s. decreasing to $S$. By
the definition (\ref{yqdef}) of $\yq$ we have $ \yq_S=\liminf_n
L^{\QQ}_{S^n}.$ Let $\set{t_1^n,\ldots, t_{m_n}^n}$ be the range
of $S^n$. Then for $A\in\FF_S\subseteq \FF_{S_n}$ we have
\begin{eqnarray*}
\EE[\yq_S\ind{A}]&=&\EE[\liminf_n\yh_{S^n}\cdot\ind{A}]\  \leq \
\liminf_n \EE[\yh_{S^n}\cdot \ind{A}]\   =\
\liminf_n\sum_{k=1}^{m_n}\EE[\yh_{t^n_k} \cdot
\ind{A\cap\set{S^n=t^n_k}}] \\ &\leq& \liminf_n\ \sum_{k=1}^{m_n}
\langle \QQ,\ind{A\cap\set{S^n=t^n_k}}\rangle  \ = \  \langle
\QQ,\ind{A}\rangle.\end{eqnarray*} Therefore, $\yq_S$ is the
density of a (countably-additive) measure dominated by $\QQ$ on
$\FF_S$, and we conclude  that $\,\yq_S\leq
\rn{(\QQ|_{\FF_S})^r}{(\PP|_{\FF_S})}$, almost surely.
\end{itemize}
\vskip-0.6cm\end{proof}

The next results, useful for the duality
treatment and interesting in their own right, introduce the notion
of {\it Fatou-convergence}, and relate it to the more familiar notion of weak
* convergence. Fatou-convergence is analogous to  a.s. convergence
in the context of RCLL-processes, and was used for example in
\cite{Kra96}, \cite{FolKra97} and \cite{8}.

\begin{defn}\lprint{FatD1}
Let $\{Y^{(n)}\}_{n\in\N}$ be a sequence of nonnegative,
$\FFF$-adapted processes with RCLL paths.
 We say that $\{Y^{(n)}\}_{n\in\N}$ {\bf Fatou-converges}
to an $\FFF$-adapted process $Y$ with RCLL-paths, if there is a
countable, dense subset $\TT$ of $[0,T]$, such that
\begin{eqnarray}
 Y_t&=&\liminf_{s\downarrow t,s\in\TT} \Bigl( \liminf_n Y^{(n)}_s
\Bigr)
\lprint{FatDef}\   = \  \limsup_{s\downarrow t,s\in\TT} \Bigl( \limsup_n
Y^{(n)}_s\Bigr)  ,\ \ \forall\, t\in [0,T]
\end{eqnarray} holds almost surely; we interpret $(\ref{FatDef})$ to mean
\ $Y_t=\lim_n Y^{(n)}_t$ a.s. for \, $t=T$. A set of nonnegative
RCLL-supermartingales is called {\bf
Fatou-closed}, if it is closed with respect to Fatou-convergence.
\end{defn}

Before stating the next proposition we need a technical result -
see Lemma 8 in \cite{Zit00}.

\begin{lem}\lprint{linf}
Let $\set{Y^{(n)}}_{n\in\N}$ be a sequence of nonnegative
RCLL-supermartingales,  Fatou-converging to a nonnegative
RCLL-supermartingale $Y$. There is a countable set $K\subseteq
[0,T)$ such that\, $Y_t=\liminf_n Y^{(n)}_t$ for all
$t\in[0,T]\setminus K$, almost surely.
\end{lem}

\smallskip
\begin{prop}\lprint{Fat2}\
Let $\mu$ be a probability measure on $[0,T]$, diffuse on $[0,T)$.
Let $\{\QQ^{(n)}\}_{n\in\N}$ be a sequence in $\DD$ with a cluster
point $\QQ^*\in\DD$, such that the sequence
$\{Y^{\QQ^{(n)}}\}_{n\in\N}$ converges, both
$(\mu\otimes\PP)$-a.e. and in the Fatou sense. Then the
Fatou-limit $Y$ coincides with the $(\mu\otimes\PP)$-limit, up to
a.e. equivalence, and both are equal to $Y^{\QQ^*}$.
\end{prop}
\begin{proof} The two limits are the same $(\mu\otimes\PP)$-a.e., by
Lemma \ref{linf}. By Proposition \ref{pryq}, there exists a
sequence $\{K_n\}_{n\in\N}$ of countable subsets of $[0,T)$, and a
$\mu$-null set $K'$, such that
\[
Y_t=\lim_n Y^{\QQ^{(n)}}_t= \lim_n L^{\QQ^{(n)}}_t, \ \
~~\text{for all} ~~ t \in [0,T]\setminus K
\]
holds almost surely, where $K\triangleq K'\cup \bigcup_{n\in\N} K_n$. By Proposition
\ref{charge}(vi), (2.4), and Proposition
\ref{pryq},  there is a $\mu$-null set $\hat{K}\supseteqq K$ such
that
\[
Y_t=Y^{\QQ^*}_t =L^{\QQ^*}_t ,\ ~~\text{for all}~~ t\in[0,T]\setminus \hat{K}
\]
holds almost surely. Since
$[0,T]\setminus \hat{K}$ is dense in $[0,T]$, the right-continuous
processes $Y$ and $Y^{\QQ^*}$ are indistinguishable.
\end{proof}

\subsection{On a point raised by  Cvitani\' c, Schachermayer and Wang}
In \cite{KraSch99}, page 6,  the authors define a set  $\YY$ of
supermartingales, which acts as an enlargement for the set of
densities of equivalent martingale measures;  they then use $\YY$
as the domain for the convex-duality approach to utility
maximization in incomplete markets. In their setup there is no
endowment after time $t=0$, no portfolio constraint, and  utility
comes from terminal wealth only. In terms of the set $\XX$ of
stochastic integrals in (\ref{XX}), the set $\YY$ of
supermartingales is  defined as
\begin{eqnarray} \nonumber \YY&\triangleq& {\Big \{ } Y\,:\, \text{$Y$ is an adapted
nonnegative RCLL process such that $Y_0\leq 1$ and}   \\
\lprint{yofs} &&  ~~~~~{\big ( } Y_tX_t {\big )}_{t\in [0,T]} \
\text{is a supermartingale for each process $X\in\XX$ } {\Big \}}\,.
\end{eqnarray}

Obviously, the elements of $\YY$ are supermartingales (just take $H=0$,
thus $X\equiv x$, in (\ref{XX})), and $\YY$ contains the set $\ym$ of
(\ref{yqdefn}) by its very definition; but except in trivial cases,
$\YY$ is a true enlargement of $\ym$. An attempt to study the
structure of $\YY$ was made in \cite{Zit00}, by establishing and
applying a generalization of the {\em Bipolar Theorem for Subsets
of} ~$\lzer_+$ (see \cite{BraSch99}); this is a non-locally-convex
version of the classical Bipolar Theorem of functional analysis.
The aforementioned generalization comes in the form of the {\em
Filtered Bipolar Theorem}, whose statement and  relevant
definitions we recall now from \cite{Zit00}:

\begin{defn} A set of $\YY$ of nonnegative $\FFF$-adapted
processes with RCLL paths, is  said to be \begin{enumerate}
\item {\bf (process-) solid}, if for each $Y\in \YY$ and each
nonincreasing $\FFF$-adapted process $B$ with RCLL paths and
$B_0\leq 1$, we have
 \  $YB\in\YY$;
\item {\bf fork-convex}, if for any
$s\in (0,T]$, any $h\in L^0_+(\FF_s)$ with $h\leq 1$ a.s., and any
$Y^{(1)},Y^{(2)},Y^{(3)}\in\YY$, the process $Y$ defined by
\[ Y_t=\left\{
\begin{array}
{cl} Y^{(1)}_t&,\  0\leq t<s \\
Y^{(1)}_s\left(
h\frac{Y^{(2)}_t}{Y^{(2)}_s}+(1-h)\frac{Y^{(3)}_t}{Y^{(3)}_s}
\right) &,\  s\leq t\leq T
\end{array}\right\}
\]
belongs to $\YY$.
\end{enumerate}
\end{defn}

\begin{defn} Let $\YY$ be a set of nonnegative, $\FFF$-adapted with RCLL paths.
The {\bf (process)-polar} of $\YY$ is the set of all nonnegative,
$\FFF$-adapted processes $X$ with RCLL paths, such that
$XY=(X_tY_t)_{t\in[0,T]}$ is a supermartingale with $(XY)_0\leq 1$
for all $Y\in \YY$.
\end{defn}

We can now state a mild extension of the main result of
\cite{Zit00}. The additional statement (last sentence of Theorem
\ref{fbt} below) follows directly from the proof of the original
version.
\begin{thm}\lprint{fbt}{\em [\bf Filtered Bipolar Theorem]\em} Let $\YY$
be a set of nonnegative,  $\FFF$-adapted processes with RCLL paths,
$Y_0\leq 1$ for each $Y\in\YY$, and with $Y_T>0$ a.s. for at least
one $Y\in\YY$. The process-bipolar
$\YY^{\times\times}=(\YY^{\times})^{\times}$ of $\YY$ is the
smallest Fatou-closed, fork-convex and solid set of $\FFF$-adapted
processes $Y$ with RCLL paths and $Y_0\leq 1$ that contains $\YY$.
Furthermore, every element of $\YY^{\times\times}$ can be obtained
as the Fatou-limit of a sequence in the solid and convex hull of $\YY$.
\end{thm}
\begin{rem}{\rm The set $\YY^{\MM}$ of (\ref{yqdefn}) is fork-convex, and its
process-bipolar is the set $\YY$ of (\ref{yofs}) (see Theorem 4 in
\cite{Zit00}). It follows immediately from Theorem \ref{fbt} that
$\YY$  is the solid and Fatou-closed hull of
$\YY^{\MM}$.}\end{rem}

The task we take on in  this subsection is to formulate and establish
formally the statement put forth by the authors in
\cite{CviSchWan99}, to the effect that

\begin{quote}{\em \ldots the idea of passing from $\MM$ to $\DD$
}(introduced in \cite{CviSchWan99}) {\em  had already been
implicitly present in \cite{KraSch99}} (disguised in the
definition of $\YY$).\end{quote} Namely, we shall show that
$\yd\subseteq\YY$, and that $\yd$ already contains all maximal
elements of $\YY$. More precisely, we have the following result.

\begin{thm} \lprint{yeyd}The set $\yd$ in (\ref{yqdefn}) is fork-convex and Fatou-closed, and
the set $\YY$ of (\ref{yofs}) is its solid hull. \end{thm}
\begin{proof}
Since $\YY$ is the process-bipolar of $\YY^{\MM}$ from
(\ref{yqdefn}), by the Fil\-te\-red Bi\-po\-lar The\-o\-rem
\ref{fbt} it is enough to prove that $\yd$ is fork-convex,
contained in $\YY$, and Fatou-closed, since $\ym$ is already contained in
$\yd$.

The fork-convexity of $\yd$ follows from its definition, from the
fork-convexity of $\ym$, and from the fact (Theorem \ref{fbt})
that every $Y\in\yd\subseteq\YY$ can be Fatou-approximated by a
sequence in $\ym$.

As for Fatou-closedness, we take a sequence
$\{\ynn\}_{n\in\N}\subseteq\yd$, Fatou-converging towards a
supermartingale $Y$. Let $\lambda$ stand for the normalized
Lebesgue measure on $[0,T]$. By Koml\' os's theorem (see
\cite{Sch86}) and the convexity of $\yd$, we can assume that
$\{\ynn\}_{n\in\N}$ converges $(\lambda \otimes \PP)-$a.e., by passing
to a sequence of convex combinations if necessary (note that this
operation preserves the Fatou-limit). Let
$\{\QQ^n\}_{n\in\N}\subseteq\DD$ be a sequence such that
$Y^{(n)}=Y^{\QQ^n}$. By the weak * compactness of $\DD$, the
sequence $\{\QQ^n\}_{n\in\N}$ possesses a cluster point $\QQ^*$.
Proposition \ref{Fat2} now yields $Y=Y^{\QQ^*}$, implying
Fatou-closedness of $\yd$.

Finally, we prove that $\yd\subseteq \YY$. Let $X\in\XX$ be such
that $X_0=1$, and let $Y\in\yd$. By the definition of $\YY$, it will be
enough to show that $YX$ is a supermartingale, and by Proposition
\ref{pryq} it is enough to prove that $L^{\QQ}X$ is a
supermartingale, where $L^{\QQ}$ is the process defined in
(\ref{ldef}). Equivalently, we have to prove  $\scl{(\qfs)^r}{ X_s
\ind{A}}\geq \scl{(\qft)^r}{X_t \ind{A}}$, for all $0\leq s<t\leq
T$, $A\in\FF_s$. For this, we may assume without loss of
generality that $X_s$ is bounded on $A$.

Recall that, for $\QQ\in\MM$, the process $X$ is a nonnegative
$\QQ$-supermartingale. By density of $\MM$ in $\DD$, we easily
conclude that $\scl{\QQ}{ X_s\ind{A}}\geq \scl{\QQ}{ \left(
X_t\wedge m \right) \ind{A}} $,  for all $\QQ\in\DD$ and $m\in
(0,\infty)$. The regular-part-operator is positive, so we have
\[  \scl{(\qfs)^r}{ X_s\ind{A}} + \scl{(\qfs)^s}{X_s\ind{A}}
\geq \scl{(\qft)^r}{ \left( X_t\wedge m\right) \ind{A}},\ \
\forall\, m\in (0,\infty).\] Proposition \ref{charge} (v)
guarantees the existence of a sequence of sets $\seq{A}$ in
$\FF_s$ such that $\PP[A_n]>1-2^{-n}$ and $(\qfs)^s(A_n)=0$. We
get
\[  \scl{(\qfs)^r}{ X_s\ind{A\cap A_n}}\geq
\scl{(\qft)^r}{ \left(X_t\wedge m\right) \ind{A\cap A_n}},\
\forall\,m\in (0,\infty),\ n\in\N,\] and the claim follows by
letting $m,n\to\infty$.
\end{proof}

For future use, we restate the result of the Theorem \ref{yeyd} in
the following terms.
\begin{cor} Every $Y\in\YY$ can be written
as $Y=Y^{\QQ}D$, where $\QQ\in\DD$, and  $D$ is a nonincreasing,
nonnegative, $\FFF$-adapted process with $D_0\leq 1$ and RCLL
paths. The process $\yq$ can be obtained as the Fatou-limit of a
sequence of martingales in $\ym$.
\end{cor}
%\begin{rem}{\rm  The previous corollary reveals an
%interesting fact about $\yd$. Since $\linfd$ is not metrizable,
%for a given $\QQ\in\DD$ it might be impossible to find a sequence
%$(\QQ^n)_{n\in\N}$ in $\MM$, weak * converging to $\QQ$. However,
%when passing to regular parts (and on a whole filtration, too), a
%variant of metrizability suddenly emerges.
%}\end{rem}

\subsection{A Characterization Of Admissible Consumption Processes} The
enlargement of the dual domain from $\MM$ to $\DD$ necessitates a
reformulation  of certain old results in the new setting. As given
in subsection \ref{FinMar1}, the definition of an admissible
consumption process is as intuitively graspable as practically
useless. To remedy this situation, we establish a
budget-constraint-characterization of admissible consumption
processes, analogous to Theorem 3.6, p. 166 in [KS98], in the
context of the
%. We remind the reader of an arbitrary but fixed cumulative
endowment process ~$x+\EN=\left(x+\EN_t\right)_{t\in [0,T]}$.

\begin{prop} \lprint{OU} A nonnegative, nondecreasing,
right-continuous and $\FFF$-adapted  process
$C$ is an admissible cumulative consumption process, if and only if
\begin{equation}\lprint{c1} \EE\left[\int_0^T Y^{\QQ}_t\,dC_t\right]\leq
x+\scl{\QQ}{\EN_T},\ \ \mbox{for all \ $\QQ\in\DD$.}\end{equation}
\end{prop}
\begin{proof}
Let $C$ be a nonnegative nondecreasing adapted right-continuous
process satisfying (\ref{c1}). For each probability measure
$\QQ\in\MM$, the process $Y^{\QQ}$ is the RCLL  modification of
the martingale $\EE[\rn{\QQ}{\PP}|\FF_t]$, $t\in [0,T]$. By virtue
of the left-continuity and existence of right-limits for the
process $t\mapsto C_{t-}$, the stochastic integral
$\, M_t\triangleq\int_0^t C_{u-}\, d\yq_u\,,~ 0 \le t \le T\,$ is a local martingale
(\cite{Pro90}, Theorem III. 17),  so we can find a non-decreasing
sequence of stopping times $\{T_n\}_{n\in\N}$ such that $\PP[T_n=
T]\to 1$ as $n\to\infty$, and the processes $M^{T_n}_{\cdot}\equiv
M_{\cdot\wedge T_n}$ are uniformly integrable martingales, for
each $n\in\N$. By the assumption (\ref{c1}) and the
integration-by-parts formula, we have
\begin{eqnarray}\nonumber \lprint{ce2} x+\scl{\QQ}{\EN_T}&\geq& \ei{T}{\yq_t\, dC_t}=
\lim_{n} \ \ei{T_n}{\yq_t\, dC_t}=\lim_n \left(
\ei{T_n}{\yq_{t-}\, dC_t}+ \sum_{s\leq T_n} \Delta \yq_s\Delta
C_s\right)\\ &=& \lim_n \left(
\EE\left[\yq_{T_n}C_{T_n}-\int_0^{T_n} C_{t-}\,
d\yq_t\right]\right) =\lim_n \ \EE_{\QQ}[C_{T_n}]=
\scl{\QQ}{C_T}.\end{eqnarray}
 Let us define \[ Z_t\triangleq\esssup_{\QQ\in\MM} \EE_{\QQ}[C_T-\EN_T|\FF_t],\
 ~~0\leq t\leq T.\] From
Theorem 2.1.1 in \cite{ElkQue95}, the process $Z$ is a
supermartingale under {\it each} $\QQ\in\MM$, with a RCLL modification.
Choose this RCLL  version for $Z$. Moreover, $Z$ is uniformly
bounded from below and $Z_0\leq x$; this is because
$\EE_{\QQ}[C_T-\EN_T]\leq x$ for every $\QQ\in\MM$, thanks to
(\ref{ce2}). Applying the Constrained Version of the Optional
Decomposition Theorem (see \cite{FolKra97}, Theorem 4.1) to $Z$,
we can assert the existence of an admissible portfolio $\hH$ and
of a nondecreasing optional process $F$ with $F_0\geq 0$, such
that
 $ Z_t=\hat{X}_t-F_t$, where $\hat{X}_t\triangleq x+\int_0^t \hH_u\,{ dS}_u.$
 On the other hand, by
the increase of $C$ we have
 \[\hat{X}_t-F_t=Z_t\geq C_t-\essinf_{\QQ\in\MM}
\EE_{\QQ}[\EN_T|\FF_t],\ \ \ \ \ t\in [0,T],\]
so that $\hat{X}_T-C_T+\EN_T\geq F_T\geq F_0\geq 0$, a.s.,
implying the admissibility of the strategy $(\hH,C)$.
\medskip

Conversely, let $C$ be an admissible consumption process; there
exists then an admissible porfolio process $H$, such that the process
$X_{\cdot}\triangleq x+\int_0^{\cdot}H_u\,dS_u$ satisfies
$X_T-C_T+\EN_T\geq 0$. By the supermartingale property of $X$
under every $\QQ\in\MM$, we conclude that $\scl{\QQ}{C_T}\leq x
+\scl{\QQ}{\EN_T}$, $\forall\, \QQ\in\MM$. Suppose first that $C$ is
uniformly bounded from above by a constant $M$, and define its
right-continuous inverse (taking values in $[0,\infty]$) by
\[ D_s=\inf\set{\,t\geq 0\,:\, C_t>s},\ \ 0\leq s<\infty.\]
For an arbitrary, but fixed $\QQ\in\DD$, by Theorem 55 in
\cite{DelMey82} and Fubini's theorem, we can write
\[ \EE\left[\int_0^T Y^{\QQ}_t\, dC_t\right]=\EE\left[\int_0^M
\yq_{D_s}\inds{D_s<\infty}\, ds\right] = \int_0^M \phi(s) \, ds,\]
where $\phi(s)=\EE[\yq_{D_s}\inds{D_s<\infty}]$. By the
supermartingale property of $\yq$ and the increase of $D$, the
function $\phi$ is nonincreasing, so we can find a countable set
$K$, dense in $[0,M]$, that contains all discontinuity points of
$\phi$. For a denumeration $\{s_k\}_{k\in\N}$ of $K$, the topology
on $\DD$ induced by the pseudometric
\[ d(\QQ_1,\QQ_2)=|\scl{\QQ_1-\QQ_2}{C_T}|+\sum_{k} 2^{-n}|
\scl{\QQ_1-\QQ_2}{\inds{D_{s_k}<\infty}}|\] is
 coarser than the weak * topology on $\DD$, so we
can find a {\em sequence} $\{\QQ^n\}_{n\in\N}\subseteq\MM$ such
that
\[\scl{\QQ^n}{C_T}\to \scl{\QQ}{C_T}\ \mbox{\ \ and\ }\
\scl{\QQ^n}{\inds{D_{s}<\infty}}\to \scl{\QQ}{\inds{D_s<\infty}},\
\ \text{as $n\to\infty$},
\] for every $s\in K$. Such choice for the sequence $\{\QQ^n\}_{n\in\N}$ implies
that $\phi^n(s)=\EE_{\QQ^n}[\inds{D_s<\infty}]$ converges to
$\scl{\QQ}{\inds{D_s<\infty}}$ for {\em every} $s$. Using again
Theorem 55 in \cite{DelMey82}, the integration-by-parts formula
from the first part of the proof, and the Dominated Convergence
Theorem, we get
\begin{eqnarray*}  x+\scl{\QQ}{\EN_T} &=&  x+\lim_n \scl{\QQ^n}{\EN_T}
 ~ \geq ~ \lim_n\scl{\QQ^n}{C_T}\  =\ \lim_n \EE\left[\int_0^T
Y^{\QQ^n}_t\, dC_t\right]=\lim_n \int_0^M \phi^n(s)\\ &=& \int_0^M
\scl{\QQ}{\inds{D_s<\infty}}\, ds.
\end{eqnarray*}
As $D_s$ is a stopping time, Proposition \ref{pryq},(b) yields
\begin{eqnarray*} \int_0^M \scl{\QQ}{\inds{D_s<\infty}}\,ds
\ \geq\  \int_0^M \EE\left[
\frac{d(\QQ|_{\FF_{D_s}})^r}{d(\PP|_{\FF_{D_s}})}
\inds{D_s<\infty}\right]\,ds \ \geq\  \int_0^M
\EE[Y^{\QQ}_{D_s}\inds{D_s<\infty}]\,ds \ =\  \EE\int_0^T
Y^{\QQ}_t\,dC_t, \end{eqnarray*} which establishes the claim.

We turn now to the case of $C$ which is not necessarily bounded.
For $M\in\N$, the truncated consumption process $C^M=C\wedge M$ is
admissible and (\ref{c1}) holds with $C$ replaced by $C\wedge M$.
Passing to the limit as $M\to\infty$ on the left-hand side of
(\ref{c1}) is justified by the increase of the trajectories of $C$ and
the Monotone Convergence Theorem.
\end{proof}

\begin{rem}{\rm  The necessity for the rather lengthy and technical proof of
this result (to be
more precise: the authors' inability to find a shorter one), stems
from two rather unpleasant facts: first, $\linfd$ is not
metrizable, and secondly, Fubini's theorem fails in the setting of
finitely-additive measures (see \cite{YosHew52}, Theorem 3.3, p.
57 for such a counterexample). }\end{rem}

%--------------------------------------------------------------
% PART V - Section 3 - The optimization problem
%--------------------------------------------------------------

\section{The optimization problem}\lprint{OptProb}

\subsection{The Preference Structure}

Apart from  external factors, such as market conditions and the
randomness of the endowment process $\EN$, it is important to
describe the agent's ``preference structure'' (or idiosyncratic
rapport with risk). We shall adopt the von Neyman-Morgenstern
utility approach to risk-aversion, and proceed to define a utility
random field
 \ $U:\Omega\times[0,T]\times \R_+\to\R$.

We shall impose no smoothness conditions in the time parameter.
Instead, we shall control the range of the marginal utility. As
seen in \cite{KraSch99}, a condition of {\em reasonable asymptotic
elasticity} (in the setting of an incomplete semimartingale market
with initial endowment only, and utility from terminal wealth) is
both necessary and sufficient for the existence of an optimal
investment policy. This is the reason for extending the notion of
asymptotic elasticity to the time-dependent case, and for restricting
our analysis to reasonably elastic utilities only. More precisely,
we have the following definition.
\begin{defn}
\lprint{DefUtility} A jointly measurable function
$U:\Omega\times[0,T]\times\R_+\to\R$ is called a {\bf (reasonably
elastic) utility random field}, if it has the following properties
(unless specified otherwise, all these properties are assumed to hold
almost surely and the argument $\omega\in\Omega$ will
consistently be suppressed):
\begin{enumerate}
\item \lprint{1} For a fixed $t\in[0,T]$, $U(t,\cdot)$ is strictly concave,
increasing and $C^1$ satisfying the so-called Inada conditions
$\pd U(t,0+)=\infty$ and $\pd U(t,\infty)=0$. In other words,
$U(t,\cdot)$ is a {\bf utility function}.
\item \lprint{2} There are continuous, strictly decreasing (nonradom)
functions $K_1:\R_+\to \R_+$ and $ K_2:\R_+\to \R_+$ such that for
all $t\in [0,T]$ and $x>0$, we have $ K_1(x)\leq \pd U(t,x) \leq
K_2(x)$, and $\limsup_{x\to\infty}\frac{K_2(x)}{K_1(x)}<\infty$.
\item \lprint{3} At $x=1$, $t\mapsto U(t,1)$ is a uniformly bounded function
of $(\omega,t)$ and $\lim_{x\to\infty} ( \essinf_{t,\omega}
U(t,x){)}
>0$.
\item \lprint{4} $U$ is {\bf reasonably elastic}, i.e., its {asymptotic
elasticity}  satisfies ${\mathrm AE}[U]<1$ a.s., where
\[{\mathrm AE}[U]:=\limsup_{x\to\infty} \left(\esssup_{t,\omega} \frac{x\pd
U(t,x)}{U(t,x)}\right).\]
\item \lprint{5} For any $x>0$, the stochastic process $U(\cdot,x)$
is $\FFF$-progressively measurable.
\end{enumerate}
\end{defn}
\begin{rem}{\rm
 Condition \ref{3} is the least restrictive - in fact, it
only serves to simplify the analysis by excluding some trivial
nuisances, as well as to have the expression ${\mathrm{AE}}[U]$ of
part \ref{4} well defined. It is an immediate consequence of
conditions \ref{2} and \ref{3} that the function $t\to U(x_0,t)$
is bounded for any $x_0>0$, a.s. Also,
 the trajectory $U(t,\infty)$ is either a bounded function of
$t$, or we have $U(t,\infty)=\infty$ for all $t$, a.s. }\end{rem}
\begin{exam}
\lprint{ex1} Let $\hat{U}:\R_+\to\R_+$ be a utility function as in
Definition \ref{DefUtility} (\ref{1}),   with $\hat{U}(\infty)>0$
and $\limsup_{x\to\infty} \frac{x\hat{U}'(x)}{\hat{U}(x)}<1.$ Let
$\psi$ be a measurable function of $[0,T]$ such that $0<\inf_{t\in
[0,T]} \psi (t)\leq \sup_{t\in [0,T]} \psi (t)<\infty$. Then it is
easy to see that $U(t,x)\triangleq \psi (t) \hat{U}(x)$ is a
reasonably elastic utility random field. In particular, this
example includes so-called {\em discounted} time-dependent utility
functions of the form $U(t,x)=e^{-\beta t} \hat{U}(x)$.
\end{exam}
\begin{exam}
\lprint{ex2} Let $U_1: [0,T]\times\R_+\to\R$ be a deterministic
utility   field with corresponding $K_1$ and $K_2$ as in
Definition \ref{DefUtility},(2). Further, let $U_2:\R_+\to\R$ be a
utility function satisfying
\[ U_2(\infty)>0,~~ \limsup_{x\to\infty} \frac{xU_2'(x)}{U_2(x)}<1\
\ \ \mbox{and}\ \ \ 0<\liminf_{x\to\infty}
\frac{U_2'(x)}{K_1(x)}\leq\limsup_{x\to\infty}
\frac{U_2'(x)}{K_1(x)}<\infty.\] One can check then the
requirements of Definition \ref{DefUtility} to see that
\[ U(t,x):=\left\{ \begin{array}{cl} U_1(t,x),& t<T\\ U_2(x), & t=T
\end{array}\right\}\] is a reasonably elastic utility
random field.
\end{exam}
\begin{exam}\lprint{stochutil} Let $U_1:[0,T]\times\R_+\to\R$ be any
deterministic reasonably elastic utility   field, and let $B_t$ be
a adapted process uniformly bounded from above and away from zero.
To model a stochastic discount factor, we define $U(t,x)\triangleq
U_1(t,B_tx)$. Such a utility random field arises when the agent
accrues utility from nominal, instead of real value of
consumption.
\end{exam}

 With a utility random field $U$ we
associate a random field $V:\Omega\times[0,T]\times \R_+\to \R$
defined by \begin{equation}\lprint{defv}
V(t,y)\triangleq\sup_{x>0} [ U(t,x)-xy],\ \ \
0<y<\infty,\end{equation}
 the {\bf conjugate of} $\,U$. We also define the random field
$I:\Omega\times[0,T]\times \R_+\to \R$, by $I(t,y)=( \pd
U(t,\cdot))^{-1}(y),$ the {\bf inverse marginal utility} of $U$.
The following proposition lists some important,  though technical,
properties of these random fields and their conjugates. They will
be used extensively in the sequel. We leave the proof to the
diligent reader.
\begin{prop}\lprint{conjuv}
 Let $U$ be a
utility random field and $V$ its conjugate.
\begin{enumerate} \item There are (deterministic) utility functions
$\uu$ and $\ou$ such that \[ \uu(x)\leq U(t,x)\leq \ou(x) \
\text{for all  $x>0$ and all $t\in[0,T]$},\ \text{a.s.}\] \item
For a given $t\in [0,T]$, the function $V(t,\cdot)$ is  finite
valued, strictly decreasing, strictly convex and continuously
differentiable.
\item The convex conjugates $\uv$ and $\ov$ of $\uu$ and $\ou$
satisfy
\[ \uv(y)\leq V(t,y)\leq \ov(y)\  \text{for all  $y>0$, and all $t\in[0,T]$\
 a.s.}\] In particular, the function $t\mapsto V(t,y)$ is uniformly
bounded, for any $y\in (0,\infty)$.
\end{enumerate}
\end{prop}
%\begin{proof}
% Claim 2. follows from well known properties of convex conjugation,
%and 3. is a consequence of 1. and 2, so we prove only Claim 1.
%Take functions $K_1$ and $K_2$ as in Definition \ref{DefUtility},
%\ref{3}. and set
%\begin{eqnarray*} \uu^*(x)&=&\essinf_{t,\omega} U(t,1)+\int_1^x
%K_1(\xi)\,d\xi,\ \mbox{ for $x\geq 1$ and }\\
%\uu^*(x)&=&\essinf_{t,\omega} U(t,1)-\int_x^1 K_2(\xi)\,d\xi\
%\mbox{ otherwise.}\end{eqnarray*} Similarly
%\begin{eqnarray*}
% \ou^*(x)&=&\esssup_{t,\omega} U(t,1)+\int_x^1 K_2(\xi)\,d\xi,\ \mbox{
%for $x\geq 1$ and }\\ \ou^*(x)&=&\esssup_{t,\omega}
%U(t,1)-\int_x^1 K_1(\xi)\,d\xi\ \mbox{ otherwise.}
%\end{eqnarray*}
%$\uu^*$ and $\ou^*$ can be smoothed out in a neighborhood
%of $1$ to yield  utility functions $\uu$ and $\ou$ with
%desired properties.
%\end{proof}
\begin{defn} \lprint{defmin}Any utility functions (i.e., strictly concave, increasing,
and continuously differentiable
functions that satisfy the Inada conditions) $\uu:\R_+\to\R$ and
$\ou:\R_+\to R$,
 such that $\uu(x)\leq U(t,x)\leq \ou(x)$ for all $x>0$ and $t\in [0,T]$,  are
called a {\bf\em minorant} and a {\bf\em majorant} of $U$,
respectively. Functions $\uv:\R_+\to\R$ and $\ov:\R_+\to\R$, that
are convex conjugates of some minorant and  majorant of $U$, are
called a {\bf\em minorant} and a {\bf\em majorant} of $V$,
respectively.
\end{defn}
\begin{rem}{\rm  It follows immediately from the definition of convex
conjugation that for any minorant and majorant $\uv$ and $\ov$ of
$V$, we have $\uv(y)\leq V(t,y)\leq \ov(y)$, for all $y>0$ and
$t\in[0,T]$. }\end{rem}
 Finally, we state a technical
result stemming from the reasonable-asymptotic-elasticity
condition; its proof is, mutatis mutandis, identical to the proof
leading to Corollary 6.3., page 994 of \cite{KraSch99}.
\begin{prop} \lprint{ae1} Let $U$ be a utility random field. If we
define the random sets
\begin{eqnarray*}
\Gamma_1&=&{\Big\{} \gm>0\, :\, \exists\, x_0>0,\,\ \forall\, t\in
[0,T],\,\ \forall\, \ld>1,\, \ \forall\, x\geq
x_0,\ \ U(t,\ld x)< \ld^{\gm} U(t,x){\Big\}}\\
\Gamma_2&=&{\Big\{} \gm>0\, :\, \exists\, x_0>0,\,\ \forall\, t\in
[0,T],\,\ \forall\, x\geq x_0, \ \ \pd U(t,x)< \gm
\frac{U(t,x)}{x}{\Big\}} \\
\Gamma_3&=&{\Big\{} \gm>0\, :\, \exists\, y_0>0,\,\ \forall\, t\in
[0,T],\,\ \forall\, 0<\rho<1,\, \ \forall\, 0<y\leq y_0,\ \
V(t,\rho y)<\rho^{-\frac{\gm}{1-\gm}}
V(t,y){\Big\}}\\
\Gamma_4&=&{\Big \{} \gm>0\, :\, \exists\, y_0>0,\,\ \forall\,
t\in [0,T],\,\ \forall\, 0<y\leq y_0,\ \ -\pd V(t,y)<
\frac{\gm}{1-\gm} \frac{V(t,y)}{y}{\Big\}},
\end{eqnarray*} then\ \
$\inf\,\Gamma_1=\inf\,\Gamma_2=\inf\,\Gamma_3=\inf\,\Gamma_4={\mathrm
AE}[U]$, a.s.
\end{prop}

\subsection{The Optimization Problem and the Main Result}

The principal task our agent is facing, is how to control
investment and consumption, in order  to achieve maximal expected
utility. At this point we have defined all notions necessary to
cast this question in precise mathematical terms.
\begin{prob}
 Let $U$ be a utility random field, $\EN$ a
cumulative  endowment process, and $\mu$ an admissible measure on
$[0,T]$ as defined in subsection \ref{FinMar1}. For an initial
capital $x>0$, we are to characterize the value function
\begin{equation} \lprint{PP} {\mathfrak
U}(x)\triangleq\sup_{c\in\AA^{\mu}(x+\EN)}
 \uc{c}\ \ \ \ \ { \mbox{\sf (Primal problem)}}.\end{equation}
\end{prob}
\begin{rem}{\rm  \lprint{nonintU} When the above $\, \mu \otimes
\PP-$integral fails to exist,
%utility $U(\cdot,\cdot)$ in the Primal problem
%turns out not to be integrable,
we set its value to be $-\infty$. This is equivalent
to the approach taken in \cite{KarShr98} where the authors
consider only consumption processes such that the negative part
$U^-(t,c(t))$ is $\, \mu \otimes \PP-$integrable. }\end{rem}

To avoid trivial situations
we adopt the following \begin{sas}\lprint{StandingAssumption}
There exists $x>0$ such that $\, {\mathfrak U}(x)<\infty$.\end{sas}
\begin{rem}{\rm
Due to the boundedness of $\EN_T$, the Standing Assumption
\ref{StandingAssumption} will hold under any conditions that
will guarantee finiteness of the value function ${\mathfrak U}$,
when $\EN_T\equiv 0$. One such a condition is given by $ 0\leq
U(t,x)\leq \kappa (1+x^{\alpha}),\ \ \forall\, x>0,\ t\in[0,T]$,
for some constants $\kappa>0$ and $\alpha\in (0,1)$. For details,
see Remark 3.9, p. 274 in \cite{KarShr98}, and compare with
\cite{KarLehShrXu91} and \cite{Xu90}. }\end{rem}

Together with the Primal problem we set up the Dual Problem with
  value function    \begin{equation}\lprint{DP}
{\mathfrak V}(y)\triangleq\inf_{\QQ\in\DD} J(y,\QQ),\
\text{where}\ J(y,\QQ)\triangleq
\left(\vq{y}{\QQ}+y\scl{\QQ}{\EN_T}\right), \ y>0\ \ \ \ \
\text{\sf (Dual problem)}.\end{equation}  It will be shown below
that the Dual problem is in fact well-posed, i.e., the integral in
its definition always exists in $\bar{\R}$. The main result of
this paper is then as follows:
\begin{thm}
\lprint{MainTheorem} Let $\,\EN= (\EN_t)_{t\in [0,T]}\,$ be a
cumulative endowment process, and $\mu$ an admissible measure.
Furthermore, let $U$ be a utility random field, let $V$ be its
conjugate, and let ${\mathfrak U}$ and ${\mathfrak V}$ be the
value functions of the Primal and the Dual Problem, respectively.
Under the Standing Assumption \ref{StandingAssumption}, the
following assertions hold:

\begin{itemize}
\item[(i)] $|{\mathfrak U}(x)|<\infty$ for all $x>0$ and $|{\mathfrak V}(y)|<\infty$ for all $y>0$,
i.e., the value functions are finite throughout their domain.
\item[(ii)] The value functions ${\mathfrak U}$ and ${\mathfrak V}$ are continuously differentiable,
${\mathfrak U}$ is strictly concave and ${\mathfrak V}$ is
strictly convex.
\item[(iii)] ${\mathfrak U}(x)=\inf_{y>0} [{\mathfrak V}(y)+xy]$, and
${\mathfrak V}(y)=\sup_{x>0} [{\mathfrak U}(x)-xy]$ for $x,y>0$,
i.e. ${\mathfrak U}$ and ${\mathfrak V}$ are convex conjugates of
each other.
\item[(iv)] The derivatives ${\mathfrak U}'$ and ${\mathfrak V}'$ of the value functions satisfy:
 \begin{eqnarray*} &&\lim_{y\to 0} -{\mathfrak V}'(y)=\lim_{x\to 0}
{\mathfrak U}'(x)\in [\inf_{\QQ\in\DD}
\scl{\QQ}{\EN_T},\sup_{\QQ\in\DD} \scl{\QQ}{\EN_T}],
 \\
&&\lim_{y\to \infty} {\mathfrak V}'(y)=\lim_{x\to \infty}
{\mathfrak U}'(x)=0.\end{eqnarray*}
\item[(v)] Both Primal
and Dual Problem have solutions $\hc^x\in\AA^{\mu}(x+\EN)$ and
$\hat{\QQ}^y\in\DD$, respectively, for all $x,y>0$. For $x>0$ and
$y>0$ related by ${\mathfrak U}'(x)=y$, we have
\[ \hat{c}^x(t)=I(t,yY^{\hat{\QQ}^y_t}),\ \ \ 0\leq t\leq T,\]
where $\hat{\QQ}^y$ is a solution to the Dual problem
corresponding to $y$. Furthermore,  $\hc^x$ is the unique optimal
consumption-rate process , and $\hat{\QQ}^y$ is determined
uniquely as far as the process $Y^{\hat{\QQ}^y}$ and the action of
$\hat{\QQ}^y$ on $\EN_T$ are concerned.
\item[(vi)] The derivative ${\mathfrak V}'(y)$ satisfies
\[ {\mathfrak V}'(y)=\scl{\hat{\QQ}^y}{\EN_T}-\eit{  Y^{\hat{\QQ}^y}_t
I(t,yY^{\hat{\QQ}^y}_t )},\] where $\hat{\QQ}^y$ is the solution
to the Dual problem corresponding to $y$.
\end{itemize}
\lprint{maintheorem} \end{thm} \begin{exam} Let $U_1$ be a utility
random field and $U_2$ a utility function. Consider the  problem
of maximizing expected utility from consumption and terminal
we\-alth
\begin{equation}\lprint{ct}
{\mathfrak U}(x):=\sup \left(\EE\left[ \int_0^T U_1(t,c(t))
{\mathrm dt} + U_2(X_T)\right]\right),\end{equation} where the
supremum is taken over all admissible investment-consumption
strategies. This problem can be regarded as a special case of our
Primal problem. Indeed, if we view the terminal wealth as being
consumed instantaneously, we can translate (\ref{ct}) into
\[ {\mathfrak U}(x)=\sup_{c\in\AA^{\mu}(x+\EN)} \uc{c}\] where
$\mu=\frac{1}{2T}\lambda+ \frac{1}{2}\delta_{\{T\}}$ ($\lambda$
denotes the Lebesgue measure on $[0,T]$) and
\[ U(t,x):=\left\{ \begin{array}{cl} 2TU_1(t,\frac{x}{2T})& t<T\\
2U_2(\frac{x}{2}) & t=T
\end{array}\right\},\]
if $U_1$ and $U_2$ satisfy the requirements of Example \ref{ex2}.
In this case $C_T-C_{T-}=\frac{1}{2} c(T)$ plays the role of
terminal wealth.\lprint{ex3}
\end{exam}

%---------------------------------------------
% PART VI - Section 4 - Examples
%---------------------------------------------

\section{Examples}\lprint{Examples}

\subsection{The It\^ o-process model}

We specialize the specifications of our model as follows:

Let $(\Omega, \FF,$ $\prf{\FF_t}, \PP)$ be a stochastic base
supporting a $d$-dimensional Brownian motion $W=\prf{W_t}$, and we
assume that $\FFF\triangleq\prf{\FF_t}$ is the augmentation of the
filtration generated by $W$. The market co\" efficients are given
by a bounded real-valued {\bf interest rate} process $r$, a
bounded {\bf appreciation rate} process $b$ taking values in
$\R^d$ and a $(d\times d)$-matrix valued {\bf volatility} process
$\sigma$. We assume that $r$, $b$ and $\sigma$ are progressively
measurable and $\sigma(t)$ is a symmetric regular matrix for each
$t$,  with all eigenvalues uniformly bounded from above and away
from zero, almost surely.

The dynamics of the money market (numeraire asset) and the stock
market is given by
\begin{equation}\lprint{model} \left\{
\begin{array}{rclc} dB_t&=&B_t r(t) dt,& B_0=1\\
dS_t&=&S_t'\left[ b(t)dt+\sigma(t)dW_t\right],&
S_0=s_0\end{array}\right.
\end{equation} where $s_0$ is a given vector in $\R^d_{++}$.
We define the {\bf market price of risk} by
\[ \theta(t)=\sigma^{-1}(t)\left[ b(t)-r(t)\ind{d}\right],\] with
$\ind{d}$ denoting the $d$-dimensional vector
$\ind{d}\triangleq(1,1,\ldots,1)'$.

 We note that the equations in (\ref{model})
specify a complete market model which, however, becomes incomplete
by introducing a cone $\KK$ of portfolio constraints, and in this
case we have (see \cite{KarLehShrXu91}, p. 712; \cite{CviKar92},
p. 777; \cite{ElkQue95}, p. 50) that the set $\YY^{\MM}$ of
(\ref{yeyd}) satisfies
\[
\YY^{\MM}\subseteq\sets{Z_{\nu}(\cdot)}{\nu\in\KKB,\, \text{such
that $Z_{\nu}(\cdot)$ is positive martingale} }.
\] Here $\KKB$ is the set of all progressively
measurable processes $\nu:[0,T]\times\Omega\to\R^{d}$, such that
\[ \int_0^T \norm{\nu(t)}^2\,dt<\infty \ \ \ \ \text{and~~ \ \ \
$\nu(t)' p\geq 0,\ \forall\,p\in\KK,\,t\in [0,T]$}\] hold almost
surely (i.e., $\nu$ takes values in the barrier cone of $-\KK$),
and
\[ Z_{\nu}(\cdot)\triangleq\exp\left( -\int_0^{\cdot} (\theta(t)+\sigma^{-1}(t)\nu(t))'\,
dW_t-\frac{1}{2}\int_0^{\cdot}
\norm{\theta(t)+\sigma^{-1}(t)\nu(t)}^2\,dt\right).\]

Let us recall also the $\sigma(\linfd,\linf)$-closure $\DD$ of the
set $\MM$ in $\linfd$, as well as the enlargement $\yd$ of
$\YY^{\MM}$ as in (\ref{yeyd}). In the following proposition we
characterize the subset $\yd_{{\rm  max}}$ of the dual domain
$\yd$, consisting of processes that are strictly positive on the
support $\mathrm{supp}\,\mu$ and {\em maximal} - i.e., not
dominated by any other process in $\yd$. We remind the reader that
$\mathrm{supp}\,\mu$ is defined to be $[0,T]$ if $\mu$ charges
$\set{T}$, and $[0,T)$ otherwise.

\smallskip
\begin{prop} \lprint{Brownian} The elements of $\yd_{{\rm  max}}$ are local martingales of
the form $\PP[ Y_t=Z_{\nu}(t),\ \forall\,
t\in\mathrm{supp}\,\mu]=1$ for some $\nu\in\KKB$.\end{prop}
\begin{proof} For simplicity, and without loss of
generality, we shall assume in this proof that the market co\"
efficient processes $r$ and $b$ are identically equal to zero, that
the volatility matrix $\sigma$ is the identity matrix, and
that $\mathrm{supp}\,\mu=[0,T]$.

 Let
$Y^{\rm  max}\in\yd_{{\rm  max}}$. The multiplicative
decomposition theorem for positive special semimartingales (see
\cite{Jac79}, Propositions 6.19 and 6.20) implies that there is
continuous local martingale $M$ with $M_0=1$, and a nonincreasing
predictable RCLL  process $D$ with $D_0=1$ and $D_T>0$ a.s., such
that $Y^{\rm  max}_t=M_tD_t$. By the martingale representation
theorem for the Brownian filtration (see \cite{KarShr91}, Theorem
3.4.15 and  Problem 3.4.16), there is a $d$-dimensional
$\FFF$-progressively measurable process $\nu$ with $\int_0^T
\norm{\nu(s)}^2\,ds<\infty$ a.s. such that
\[ M_t=\exp\left(-\int_0^t
\nu(s)'\,dW(s)-\frac{1}{2} \int_0^t \norm{\nu(s)}^2\,ds\right),\ \
0\leq t\leq T.
\] For any admissible trading strategy $H$  and $x>0$ such that
\[ X^{x,H}_t \triangleq\, x+\int_0^t H(s)'\,dW(s)\geq 0,\
\text{ \ $\forall\,t\in [0,T]$} \] holds almost surely, the
process $YX^{x,H}$ is a supermartingale by Theorem \ref{yeyd}.

By It\^ o's formula (e.g. \cite{Pro90}, Section II.7) we have\  $
d(Y_t\xpi_t)=\xpi_t\, dY_t+Y_{t-}\, d\xpi_t+\, d[\xpi,Y]_t $\ and
$ dY_t=M_t\, dD_t+D_{t-}\, dM_t+\, d[M,D]_t.$ Since $M$ is
continuous and $D$ is predictable and of finite variation,
$[M,D]_t\equiv M_0D_0$, so $\lprint{pi3} \, dY_t=M_t\,
dD_t-D_{t-}M_t\nu(t) \, dW_t,$ because $dM_t=-M_t\nu(t)'\, dW_t$.
Furthermore, $ d\xpi_t=H_t'\, dW_t$, so
$d[\xpi,Y]_t=-D_{t-}M_tH_t'\nu(t)\, dt.$ It follows that
\begin{equation}\lprint{setmeas}
d(Y_t\xpi_t)=L_t+M_t\left[\xpi_t\,dD_t-D_{t-}H_t'\nu(t)\,dt\right],\end{equation}
where $L$ is a local martingale.

Now we prove that $\nu\in\KKB$. To do that, let us assume {\em per
contra} that $\nu$ fails to satisfy the relation:
\begin{equation}\lprint{nupe}\nu(t)' p\geq 0 \ \text{ for all $p\in\KK$,
\ $\lambda\otimes\PP$-a.s.}\end{equation} Then, we can find a
constant $\eps>0$, a predictable set $A$ such that
$(\lambda\otimes\PP)(A)> 0$, and a bounded predictable process
$\hat{H}$ taking values in $\KK$, such that $\hat{H}=0$ off $A$
and \begin{equation} \lprint{setpo} D_{t-}\nu(t)' \hat{H}_t\leq
-\eps\ \text{\ \ on}\
 A.\end{equation} We can also assume that $\norm{\hat{H}_t}=1$ on $A$,
$(\lambda\otimes\PP)$-a.s. For any $x>0$, we define $S^x$ to be
the first hitting time of the origin for the continuous process
$X^{x,\hat{H}}$. Also, for $x>0$ we define
$H^x_t\triangleq\hat{H}_t\ind{\sint{0,S^x}}(t)$, so that
$X^{x,H^x}_t\geq 0$ for all $t\in [0,T]$, a.s. Now we have all the
ingredients to define a family of signed measures
$\{\varphi_x\}_{x>0}$, given by
\begin{equation} \lprint{negmea}
\varphi_x(B)\triangleq\EE\left(\int_0^T\ind{B}(t)X^{x,H^x}_t\,dD_t+\eps
\int_0^T\ind{B}(t)\,dt\right),\end{equation} on the
$\FFF$-predictable  subsets of $[0,T]\otimes\Omega$. By the
supermartingale property of $YX^{x,H^x}$, relations
(\ref{setmeas}) and  (\ref{setpo}), and the strict positivity of
the process $M$, we have that  $\varphi_x(B)\leq 0$, for any $x>0$
and any $\FFF$-predictable set $B\subseteq A\cap \sint{0,S^x}$.
Due to the fact that $H^x$ is zero off $A$, $A\cap \sint{0,S^x}$
is still of positive $(\mu\otimes\PP)$-measure. By Theorem 2.1 of
\cite{DelSch95}), there exists an $\FFF$-predictable process
$g:[0,T]\times\Omega\to\R$ and an $\FFF$-predictable set
$N\subseteq [0,T]\times\Omega$ such that
\begin{equation}\lprint{SchAbs} D_t=\int_0^t g(u)\,du+\int_0^t \ind{N}(s)\,dD_u\ \ \
\text{and}\ \ \int_0^t\ind{N}(u)\,du=0\ \ \text{for all $t\in
[0,T]$}\end{equation} hold almost surely,  and $\int_0^T g(u)\,
du\leq D_T\leq 1$, a.s. From the definition (\ref{negmea}) of
$\varphi_x$ and the decomposition (\ref{SchAbs}), for any $x>0$
and any predictable $B\subseteq A\cap \sint{0,S^x}\setminus N$, we
have \ea { \lprint{fri} 0&\geq& \varphi_x(B)=
\EE\left(\int_0^T\left(X^{x,H^x}_tg(t)+\eps\right)\ind{B}(t)\,dt\right).}
The equation (\ref{SchAbs}) states that $(\lambda\otimes\PP)(N)=0$
for all $x>0$, so (\ref{fri}) implies that
$X^{x,H^x}_tg(t)+\eps\leq 0$ holds $(\lambda\otimes\PP)$-a.e. on
$A\cap \sint{0,S^x}$, for any $x>0$. We observe that the
right-continuous inverse $Q^{-1}$ of the process $Q$ given by
\begin{equation}\lprint{var} Q_t\triangleq\int_0^t
\ind{A}(s)\,ds=\int_0^t
\norm{\hat{H}_s}^2\,ds=[X^{0,\hat{H}},X^{0,\hat{H}}]_t,\ \ \ 0\leq
t \leq T \end{equation} is  a random time-change which turns the
process $X^{0,\hat{H}}$ into a Brownian motion ${\xi}_s\triangleq
X^{0,\hat{H}}_{Q^{-1}_s}$ on the stochastic interval ${\mathbb
S}\triangleq\sinto{0,Q_T}$ (see Theorem 4.6, p. 174 and Problem
4.7, p. 175 in \cite{KarShr91}).  Let $f(s)$ be the composite
process $g(Q^{-1}_s)$, and let $R_x=Q_{S^x}$ be the hitting time
of $-x$ by the Brownian motion $\xi$. Thus, for any $x>0$ and any
$\FFF$-predictable set $B\subseteq{\mathbb S}\cap\sint{0,R^x}$, we
have
\begin{equation} \lprint{absint}
1\geq -\int_0^T \ind{A\cap\sint{0,S^x}}(u)g(u)\,du\geq
-\int_0^{Q_T} \ind{B}(v) f(v)\, dv\geq
\eps\int_0^{Q_T}\ind{B}(v)\frac{1}{x+{\xi}_v}\,dv,\ \text{a.s.}
\end{equation} The relation (\ref{absint}) implies that
$x+B_s\geq \eps$, $(\lambda\otimes\PP)$-a.e. on ${\mathbb
S}\cap\sint{0,R^x}$. This is in contradiction with the fact that
$\PP(x+{\xi}_{R_x}=0)>0$ and, for small enough $x$, $\PP(R_x\in
{\mathbb S})>0$.

Therefore, the relation (\ref{nupe}) holds, and we know that the
process $M$   dominates $Y^{max}$. By truncation, $M$ can be
obtained as the Fatou-limit of a sequence of martingales in
$\YY^{\MM}$, so by Theorem \ref{fbt}, $M\in \YY$. Theorem
\ref{yeyd} states that $M$ is dominated by an element of $\yd$.
Since $M$ is a local martingale with $M_0=1$, and all elements
$Y\in\yd$ are supermartingales with $Y_0\leq 1$, we can find a
sequence $\{T_n\}_{n\in\N}$ of stopping times that reduces $M$,
and use it to conclude that $M\in\yd_{{\rm max}}\subseteq\yd$ and
$Y^{\rm max}=M$.
\end{proof}

Because of the fact that the optimal solution of the dual problem
must be positive on the $\mathrm{supp}\,\mu$, we have the
following:

\begin{cor} \lprint{BrownCor} In the setting  of an It\^ o-process
market, the primal problem admits a unique solution,
\begin{equation}\lprint{49} c(t)=I(t,y Z_{\nu}(t)D_t), \ \ 0\leq t\leq
T,\end{equation} for some constant $y>0$, and some predictable
process $\nu$ that takes values in the barrier cone of $-\KK$ and
satisfies $\int_0^T \norm{\nu(s)}^2\,ds<\infty$ a.s., and some
positive, nonincreasing and $\FF$-predictable process $D$ with
$D_0\leq 1$. Both processes $Z_{\nu}$ and $Z_{\nu}D$ are in $\yd$.
\end{cor}
\begin{rem}
When the market is complete, or, more generally, when the terminal
value of the endowment process is ``attainable'' (i.e., $x+\EN_T=X_T$
for some $X \in {\mathcal X}$ as in (2.2), then the dual
objective function $\QQ\mapsto J(y,\QQ)$ of (\ref{DP}) is monotone
in $\yq$ and thus the optimal solution takes the form \
$c(t)=I(t,y Z_{\nu}(t)), \ \ 0\leq t\leq T$, with $D\equiv 1$ in (\ref{49}).
\end{rem}

\subsection{Optimal Consumption of a Random Endowment}

In this example we consider a situation in which the agent must
optimally distribute an unknown future endowment without any
possibility of hedging the uncertainty in a financial market. This
problem was studied by Lakner and Slud in \cite{LakSlu91} in a
point-process setting. We shall consider the following version of it:
\begin{prob}\lprint{lakner} Let $(\Omega,\FF,\prf{\FF},\PP)$ be a
filtered probability space satisfying the usual hypotheses, and
let $\eps(\cdot)$ be a nonnegative progressively measurable
process such that $\,\EN_T=\int_0^T \eps(t)\, dt\,$ is uniformly
bounded from above and away from the origin. With $U$, a given
utility function, the question is to find a progressively
measurable, nonnegative consumption-rate process $c(\cdot)$ satisfying \
$\int_0^T c(t)\,dt<\infty$ a.s. - so as to maximize the expected
utility $ \EE\int_0^T U(c(t))\, dt,$ subject to the constraint
\begin{equation}\lprint{constr1} \int_0^T c(t)\, dt\leq
\int_0^T \eps(t)\,dt\ \ \text{a.s.}\end{equation}
\end{prob}

The following theorem was proved in \cite{LakSlu91}. As usual,
$I(\cdot)$ will denote the inverse marginal utility, i.e.
$I(y)=(U')^{-1}(y)$, for $0<y<\infty$. We include a proof for the
reader's convenience.
\begin{thm} Suppose there exists a positive $\FFF$-martingale $Y$
such that \begin{equation} \lprint{constr}\int_0^T I(Y_t)\,dt\ =\
\int_0^T\eps(t)\,dt\, ,\ \ \ \ \ \text{a.s.}\end{equation} Then  an optimal
consumption process is given by
\begin{equation}\lprint{hatc} \hat{c}(t)=I(Y_t),\ \ 0\leq t\leq T.\end{equation} \lprint{Lak1}
\vskip-0.8cm
\end{thm}
\begin{proof}
>From the inequality $U(I(y))\geq U(c)+yI(y)-yc$, valid for $y>0$,
and $c>0$, we obtain
\[ U(I(Y_t))\geq U(c(t))+Y_t I(Y_t)-Y_t c(t),\] for every
positive, adapted process $\{ c(t),\,t\in[0,T]\}$. Therefore,
\[ \EE\int_0^T U(\hat{c}(t))\,dt\ \geq \EE\int_0^T U(c(t))\,
dt\,+\, \EE\int_0^T Y_tI(Y_t)\,dt\,-\, \EE\int_0^T Y_t c(t)\,dt,\]
and the optimality of the process $\hat{c}$ in (\ref{hatc}),
amongst those that satisfy (\ref{constr1}), will follow once we
have shown that this latter constraint implies
\begin{equation} \lprint{four}\EE\int_0^T Y_t I(Y_t)\,dt\ \geq\ \EE\int_0^T Y_t
c(t)\,dt.\end{equation} To do that, it suffices to introduce the
probability measure
$\tilde{\PP}(A)\triangleq\frac{1}{y}\EE[Y_T\ind{A}]$ for
$A\in\FF_T$, where $y=\EE [Y_0]\in (0,\infty)$. This measure is
equivalent to $\PP$, and thus the martingale property of $Y$, (\ref{constr1}) and (\ref{constr})
lead to
\[ \EE\int_0^T Y_t I(Y_t)\,dt\ = \ y\,\tilde{\EE}\int_0^T I(Y_t)\,dt\
= \ y\, \tilde{\EE}\int_0^T \eps(t)\,dt\ = \ \EE\int_0^T
Y_t\, \eps(t)\,dt \geq = \ \EE\int_0^T
Y_t\, c(t)\,dt,\] which is (\ref{four}).
\end{proof}

We prove the following existence result, which is a partial
converse of Theorem \ref{Lak1}:

\begin{prop} When the utility function $U(\cdot)$ satisfies the ``reasonable asymptotic
 elasticity'' condition of Definition \ref{DefUtility} (4), the optimization Problem \ref{lakner} has a unique
solution which is of the form $\hat{c}(t)=I(Y_t),\ 0\leq t \leq
T$,  for some positive, RCLL supermartingale $Y$; this process
satisfies
\begin{equation}\int_0^T
I(Y_t)\,dt=\int_0^T\eps(t)\,dt\,,\ \
\text{a.s.}\lprint{equal}\end{equation}
\end{prop}
\begin{proof}
We note first that Problem \ref{lakner} is a special case of our
Primal problem with a one-dimensional {\em ``stock price''}
process $S_t\equiv 1$ and  trivial bond-price process $B_t\equiv
1$. In this case {\em all} measures equivalent to $\PP$ are
equivalent supermartingale measures, and by Theorem \ref{yeyd} any
RCLL-supermartingale $Y$ with $Y_0\leq 1$ is in $\YY$. By the Main
Theorem \ref{MainTheorem}, the unique optimal consumption-rate
process is given by $ \hat{c}(t)=I(y Y^{\hq}_t),\ \ 0\leq t\leq
T,$ for some $y>0$ and some $\hat{\QQ}\in\DD$. To finish the proof
we define $Y_t=yY^{\hq}_t$, and note that   Proposition \ref{OU}
implies that
\begin{equation}\int_0^T \hc(t)\,dt\ \ \leq\ \  \int_0^T \eps(t)\,dt \,,\
\ \  \text{a.s.}\lprint{domin}\end{equation} because every measure
equivalent to $\PP$ is in $\MM$. From the Main Theorem
\ref{maintheorem} (v) and (vi), it follows that, for the optimal
solution $\hq\in\DD$ of the dual problem, we have
\begin{equation} \lprint{ff}
\EE\int_0^T Y^{\hq}_T \hc(t)\, dt \,=\,
\sclb{\hq}{\int_0^T \hc(t)\,dt} \, \geq \,
\sclb{\hq}{\int_0^T\eps(t)\,dt}\,=\,
\EE\int_0^T Y^{\hq}_T \eps(t)\, dt\,
.\end{equation} The random variable
$\,Y^{\hat{\QQ}}_T = L^{\hat{\QQ}}_T = d(\hat{\QQ})^r/d{\PP}\,$ is strictly positive, so the equation
(\ref{equal}) follows from (\ref{domin}) and (\ref{ff}).
\end{proof}

%---------------------------------------------------
% PART VII - Appendix A - proof of the main theorem
%---------------------------------------------------

\appendix
\section{Proof of the main theorem \ref{maintheorem}}\lprint{ProofA}

In this part we state and prove a number of results leading to the
proof of our Main Theorem \ref{maintheorem}. To simplify the
notation we do not relabel the indices when passing to a
subsequence.

\subsection{Existence in the Dual Problem} We study the dual
problem first. In this subsection we point out some properties of
the dual objective function and establish the existence of
$\,\hat{\QQ} \in {\mathcal D}\,$ which is optimal in the dual
problem of (3.3). The
negative part $\max\{0,-V\}$ of the random field $V$ will be
denoted by $V^-(\cdot)$. Our first result establishes a
lower-semicontinuity property for the nonlinear part of the dual
objective function. We remind the reader that $V$ is the convex
cunjugate of $U$ introduced in (\ref{defv}).
\begin{lem} \lprint{UIV} For $y>0$, the family
of random processes $\{V^-(\cdot,y\yq_{\cdot})\,:\, \QQ\in\DD\}$
is uniformly integrable with respect to the product measure
$(\mu\otimes\PP)$ on $[0,T]\times\Omega$. Furthermore, the lower-
semicontinuity relation
\begin{eqnarray}\lprint{Vusc}
\vq{y}{\QQ}\,\leq\, \liminf_n \,\vq{y}{\QQ^{(n)}}
\end{eqnarray}
holds for all sequences $\{\QQ^{(n)}\}_{n\in\N}\subseteq\DD$ such
that $\yqn$ converges to a RCLL  supermartingale $\yq$,
$(\mu\otimes\PP)$-a.e.
\end{lem}
\begin{proof}
Let $\uv(\cdot)$ be a minorant of $V(\cdot,\cdot)$, as introduced
in Definition \ref{defmin}. We define $\varphi :\R_+\to\R_+$ to be
the right-continuous inverse of $\uv^-(\cdot)$, i.e. $ \varphi
(x)\triangleq\inf\set{y\geq 0\,:\, \uv^- (y)< x},\ \text{for
$x\geq0$}.$ Suppose first that $\varphi(x)$ is finite for all
$x\geq 0$. Then, by L'H\^ opital's rule,
\[ \lim_{x\to\infty} \frac{\varphi (x)}{x}=
\lim_{x\to\infty} \varphi '(x) =\lim_{y\to\infty}
\frac{1}{(\uv^-)^{'}(y)}=\infty.\] The family $\set{\varphi(V^-
(\cdot,y\yq_{\cdot}))\,:\,\QQ\in\DD}$ is bounded in
$L^1(\mu\otimes\PP)$, because
\begin{eqnarray*} \eit{\varphi (V^-(t,y\yq_t))}\ \leq\
\eit{\varphi (\uv^-(y\yq_t))}\ \leq\  \varphi (0)+\eit{y\yq_t}\
\leq\  \varphi (0)+y.
\end{eqnarray*} Thus, by the thorem of de la Vall\' e Poussin (see \cite{Shi96}, Lemma II.6.3.
p. 190), the family of random variables
 $\set{\varphi (V^- (\cdot,y\yq_t))\,:\,\QQ\in\DD}$ is   uniformly
integrable. If $\varphi(x)=\infty$, for some $x>0$, then
$\uv^-(\cdot)$ is a bounded function and uniform integrability
follows readily.

Let $\{\QQ^{(n)}\}_{n\in\N}\subseteq\DD$ be a sequence such that
$\{\yqn\}_{n\in\N}$ converges to a RCLL-supermartingale $\yq$,
$(\mu\otimes\PP)$-a.e. By uniform integrability we have that
\begin{equation}\lprint{s1}
\eit{V^- (t,y\yqn_t)}\longrightarrow\eit{V^- (t,y\yq_t)},\
\text{as $n\to\infty$}.
\end{equation}
As for the positive parts, Fatou's lemma gives that
\begin{equation}\lprint{s2}
 \liminf_n \eit{V^+ (t,y\yqn_t)}\ \geq\  \eit{ V^+(t,y\yq_t)}.\end{equation}
 The claim now follows from (\ref{s1}) and (\ref{s2}).
\end{proof}

The following result establishes the existence of a solution to
the dual problem.
\begin{prop} \lprint{exsol} For each $y>0$ such that ${\mathfrak V}(y)<\infty$, there is
$\hq\in\DD$
 such that \[{\mathfrak V}(y)\ = \ J(y,\hq)\  \ =\vq{y}{\hq} +y\scl{\hq}{\EN_T}.\]
\end{prop}
\begin{proof}
We fix $y>0$ and let $\qnn$ be a minimizing sequence for
$J(y,\cdot)$. We first assume that the sequence
$\{\scl{\qn}{\EN_T}\}_{n\in\N}$ converges in $\R$. This can be
justified by extracting a subsequence if necessary. By Lemma 5.2
in \cite{FolKra97} we can find a sequence of convex combinations
$\{\yqn_{\cdot}\}_{n\in\N}$ and a RCLL-supermartingale $Y$ such
that $\{\yqn_{\cdot}\}_{n\in\N}$ converges towards $Y$ in the
Fatou sense. Because of boundedness in $\lone(\mu\otimes\PP)$,
thanks to Koml\' os's theorem we can pass to a further sequence of
convex combinations to achieve convergence $(\mu\otimes\PP)$-a.e.
By Proposition \ref{Fat2}, the   limit is still $Y$.
 Because of the
convexity of $V(t,\cdot)$ and the convergence of the sequence
$\{\scl{\qn}{\EN_T}\}_{n\in\N}$, passing to convex combinations
preserves the property of being a minimizing sequence. By
Proposition \ref{Fat2}, the limit $Y$ is of the form $Y^{\hq}$ for
some (and then every) cluster point $\hq$ of $\qnn$; the existence
of such a cluster point is guaranteed by Alaoglu's theorem.
Invoking Lemma $\ref{UIV}$ establishes the claim of the
proposition.
\end{proof}

\subsection{Conjugacy and finiteness of ${\mathfrak U}(\cdot)$ and ${\mathfrak V}(\cdot)$}

The next step is to establish a conjugacy relation between
${\mathfrak U}(\cdot)$ and ${\mathfrak V}(\cdot)$. The most
important tool in this endeavor is the Minimax Theorem.
\begin{lem} \lprint{conj} The function ${\mathfrak V}(\cdot)$ is the convex conjugate
of ${\mathfrak U}(\cdot)$, i.e.
\[{\mathfrak V}(y)=\sup_{x>0} [{\mathfrak U}(x)-xy]\ \ \text{ for $y>0$.}\]
\end{lem}
\begin{proof}
For fixed $y\in (0,\infty)$ and $n\in\N$, let $\SS_n$ denote the set of all
nonnegative, progressively measurable processes
$c:[0,T]\times\Omega\to [0,n]$. The sets
$\SS_n$ can be viewed as a closed subsets of  balls in
$\linf(\mu\otimes\PP)$. Thanks to the concavity of $U(t,\cdot)$,
the compactness of $\SS_n$ (by Alaoglu's theorem; see
\cite{Woj96}, Theorem 2.A.9), and the convexity of $\DD$, we can
use the Minimax Theorem (see \cite{Str85}, Theorem 45.8 and its
corollaries) to obtain
\begin{eqnarray}\lprint{MT0}\nonumber
\sup_{c\in\SS_n}\inf_{\QQ\in\DD} \left( \EE\int_0^T \left(U(t,
c(t))-y\yq_t c(t)\right)
\,\mu(\dd{t})+y\scl{\QQ}{\EN_T}\right)=\\
\inf_{\QQ\in\DD}\sup_{c\in\SS_n} \left( \EE\int_0^T \left(U(t,
c(t))-y\yq_t c(t)\right)\,\mu(\dd{t}) +y\scl{\QQ}{\EN_T}\right),
\end{eqnarray} for any $n\in\N$, $y>0$.
Proposition \ref{OU} guarantees that $\cup_{x>0}
\AA^{\mu}(x+\EN)=\cup_{x>0}(\AA^{\mu})'(x+\EN)$ where \[
(\AA^{\mu})'(x+\EN)\triangleq\set{c\in \AA^{\mu}(x+\EN):\, \,
\sup_{\QQ\in\DD} \left(\EE\int_0^T c(t)
\yq_t\,\mu(dt)-\scl{\QQ}{\EN_T}\right)=x}.\] Thus, by pointwise
approximation of elements of $\ \cup_{x>0}(\AA^{\mu})'(x+\EN)\ $
by elements of $\cup_{n\in\N}\SS_n$, we obtain
\begin{eqnarray}\nonumber&&
\lim_{n\rightarrow\infty}\sup_{c\in\SS_n} \inf_{\QQ\in\DD} \left(
\EE\int_0^T \left(U(t, c(t))-y\yq_tc(t)\right)\,\mu(\dd{t})
+y\scl{\QQ}{\EN_T}\right)=\\ &=&\sup_{x>0}
\sup_{c\in(\AA^{\mu})'(x+\EN)} \EE\left[\int_0^T \left(U(t,
c(t))-xy\right)\,\mu(\dd{t})\right]= \sup_{x>0}[{\mathfrak
U}(x)-xy]. \lprint{cs}
\end{eqnarray}
We define $V^{(n)}(t,y)\triangleq\sup_{0<x\leq n}[U(t,x)-xy]$, and
the pointwise maximization yields
\begin{eqnarray} \nonumber && \inf_{\QQ\in\DD}\sup_{c\in\SS_n}
\left( \EE\int_0^T \left(U(t, c(t))-y\yq_tc(t)\right)\,\mu(\dd{t})
+y\scl{\QQ}{\EN_T}\right)\\&=& \inf_{\QQ\in\DD} \left( \EE\int_0^T
V^{(n)}(t,y\yq_t)\,\mu(\dd{t}) +y\scl{\QQ}{\EN_T}\right)\triangleq
{\mathfrak V}^{(n)}(y) \lprint{MT3}\end{eqnarray} From
(\ref{MT0}), (\ref{cs}) and (\ref{MT3}) we conclude that  $\lim_n
{\mathfrak V}^{(n)}(y)=\sup_{x>0} [{\mathfrak U}(x)-xy]$. To prove
the claim of the lemma it is enough to show that
$\lim_{n\rightarrow \infty} {\mathfrak V}^{(n)}(y)\geq {\mathfrak
V}(y)$, since ${\mathfrak V}^{(n)}(y)\leq {\mathfrak V}(y)$ holds
for all $y>0$, $n\in\N$. For a fixed $y>0$, let $\qnn\subseteq\DD$
be a sequence such that
\[\lim_{n\rightarrow\infty}\left( \EE\int_0^T
V^{(n)}(t,yY^{\qn}_t)\,\mu(\dd{t}) +y\scl{\qn}{\EN_T}\right)
=\lim_{n\rightarrow\infty}{\mathfrak V}^{(n)}(y).\] Using the
construction from Lemma \ref{UIV} we can assume that
$\scl{\qn}{\EN_T}\to\scl{\QQ^*}{\EN_T}$ and that $Y^{\qn}\to
Y^{\QQ^*}$ as $n\to\infty$, both in the $(\mu\otimes\PP)$-a.e. and
in the Fatou sense, where $\QQ^*$ is a cluster point of $\qnn$.

 Let
$\ou(\cdot)$ be a majorant of $U$, and $\ov(\cdot)$ its conjugate.
Then it is easy to see that \[V^{(n)}(t,y)\leq
\ov^{(n)}(y):=\sup_{0<x\leq n} [\ou(x)-xy],\ \text{for all $t\in
[0,T]$}\] and $\ov^{(n)}(y)=\ov(y)$ for $y\geq
\overline{I}(1)\geq\overline{I}(n)$ where $\overline{I}(y):=
(\ou'(\cdot))^{-1}(y)$. The argument from Lemma \ref{UIV} takes
care of  the uniform integrability of the sequence of processes
$\{V^{(n)}(\cdot,Y^{\qn}_{\cdot})^-\}_{n\in\N}$ as well as of the
following chain of inequalities
\begin{eqnarray*}&&\lim_{n\rightarrow\infty}\left( \EE \int_0^T
V^{(n)}(t,Y^{\qn}_t)\,\mu(\dd{t}) +y\scl{\qn}{\EN_T}\right) \
\geq\  \left( \EE \int_0^T V(t,Y^{\QQ^*}_t)\,\mu(\dd{t})
+y\scl{\QQ^*}{\EN_T}\right)\geq {\mathfrak V}(y),
\end{eqnarray*}
 settling the claim of the lemma.
\end{proof}
\begin{rem}{\rm { \lprint{finvy} It is a consequence of the decrease of
${\mathfrak V}(\cdot)$ and
 the preservation of
properness in the conjugacy relation (see \cite{Roc70}, Theorem
12.2, p. 104 )\lprint{Rocref} that the Standing Assumption
\ref{StandingAssumption} implies the existence of $y_0>0$ such
that ${\mathfrak V}(y)<\infty$ for $y>y_0$. Furthermore, the
strict convexity of $V(t,\cdot)$ allows us to denote by $\hqy$ the
unique (as far as its action on $\EN_T$ and the corresponding
supermartingale $\yqy$ are concerned) minimizer of the dual
problem for $y$ such that ${\mathfrak V}(y)<\infty$. } }\end{rem}

\begin{lem} \lprint{finv} ${\mathfrak V}(y)\in (-\infty,\infty)\ \text{for all
$y>0$.}$ \end{lem}
 \begin{proof}
 Let $\uu(\cdot)$ be a minorant of $U(\cdot,\cdot)$.
 $\uu(\cdot)$ is a utility function and the convex
conjugate $\uv(\cdot)$ of $\uu(\cdot)$ satisfies $\uv(y)\leq
V(t,y)$ for all $t$.
 Let $\rho=\norm{\EN_T}_{\linf}$. By the convexity of $\uv(\cdot)$
and Jensen's inequality, we have
 \begin{eqnarray}\nonumber\lprint{Jens}
 {\mathfrak V}(y)&=& \inf_{\QQ\in\DD} \left(\eit{V(t,y\yq_t)}+y\scl{\QQ}{\EN_T}\right)
 \ \geq\ \inf_{\QQ\in\DD} \eit{\uv(y\yq_t)}
 \\ &\geq&  \inf_{\QQ\in\DD} \uv\left(\eit{y\yq_t}\right)\  \geq\  \uv(y)>-\infty.
 \end{eqnarray}
 To prove that ${\mathfrak V}(y)$ is finite,  we first choose $y>0$
 such that ${\mathfrak V}(y)<\infty$ -- its existence is guaranteed by Remark \ref{finvy}.
For some $\gamma\in\Gamma_3\cap [{\mathrm AE}[U],1)$ a.s,  and
some $0<\rho<1$,  Proposition \ref{ae1} implies that there exists
$y_0>0$ such that
\[ V(t,\rho y)\leq C\, V(t,y),\ \text{for $y\leq y_0$}\] where
$C=\rho^{-\gamma/(1-\gamma)}$. By Proposition \ref{exsol} there is
$\hq^{y}\in\DD$ such that
 ${\mathfrak V}(y)=\eit{V(t,yY^{\hq^{y}}_t)}$, so
 \begin{eqnarray*}
 {\mathfrak V}(\rho y)&\leq& \eit{V(t,\rho y \yqy_t)}\\
 &=& \eit{V(t,\rho y \yqy_t)\inds{\rho y \yqy_t>y_0}} +
 \eit{V(t,\rho y \yqy_t)\inds{\rho y \yqy_t\leq y_0}}\\
 &\leq& \sup_t V(t,y_0) + C\eit{V(t,y \yqy_t)\inds{\rho y
 \yqy_t\leq y_0}}<\infty.
 \end{eqnarray*} We conclude that ${\mathfrak V}(y)<\infty$ for all
 $y>0$, due to the decrease of ${\mathfrak V}(\cdot)$.
 \end{proof}

 Having established the existence and essential uniqueness of the
 solution, and the finiteness of the value function for the dual
 problem, we can apply ideas from the calculus of variations to
 obtain the following:
\begin{lem} \lprint{difv} For each $y>0$ and
each $\QQ\in\DD$ we have
\[ \eit{(Y^{\QQ}_t-\yqy_t)I(t,y\yqy_t)}+\scl{\hqy-\QQ}{\EN_T}\ \leq\  0,\]
where $\hq^y$ is the optimal solution to the dual problem of (3.3) (as in
Proposition \ref{exsol} and Remark 10).
\end{lem}
\begin{proof} For $y>0$, $\eps\in (0,1)$ and
$\qe=(1-\eps)\hqy+\eps\QQ$, the optimality of $\hq^y$ implies
\begin{eqnarray*}
0&\leq&
\eit{\left(V(t,y\yqe_t)-V(t,y\yqy_t)\right)}+y\scl{\qe-\hqy}{\EN_T}\\
&\leq& \eit{y(\yqy_t-\yqe_t)I(t,y\yqe_t)}+y\scl{\qe-\hq}{\EN_T}\\
&=& \eps y \left(
\eit{(\yqy_t-\yq_t)I(t,y\yqe_t)}+\scl{\hqy-\QQ}{\EN_t}\right).
\end{eqnarray*}
Since
\[ \left( (\yq_t-\yqy_t)I(t,y\yqe_t)\right)^-\ \leq\  \yqy_t I(t,y\yqe_t)
\ \leq\  \yqy_t I(t,y(1-\eps)\yqy_t),\] we can follow the same
reasoning as in Lemma \ref{finv} to show that the last term is
dominated by an random process  on $\Omega\times [0,T]$ which is
$(\mu\otimes\PP)$-integrable. Now we can let $\eps\to 0$ and apply
Fatou's lemma, to obtain the stated inequality.
\end{proof}

\subsection{Differentiability of the value functions}
We turn our attention the the differentiability properties of the
value functions.
\begin{prop} \lprint{vprime} The dual value function ${\mathfrak V}(\cdot)$ is strictly convex and continuously
differentiable on $\R_+$; its derivative is given by
\[ {\mathfrak V}'(y)=\scl{\hqy}{\EN_T}-\eit{\yqy I(t,y\yqy)}.\]
\end{prop}
\begin{proof}
The fact that ${\mathfrak V}(\cdot)$ is strictly convex follows
from the strict convexity of $V(t,\cdot)$. Therefore, to show that
${\mathfrak V}(\cdot)$ is continuously differentiable, it is
enough (by convexity) to show that its derivative exists everywhere
on $(0,\infty)$. We start by fixing $y>0$, and defining the function
\[ h(z)\triangleq\eit{V(t,z\yqy_t)}+z\scl{\hqy}{\EN_T},\ \  z>0\]
This function is convex and, by definition of the optimal
solution $\hqy$ of the dual problem, we have $h(z)\geq {\mathfrak
V}(z)$ for all $z>0$ and $h(y)={\mathfrak V}(y)$. Again by
convexity, we obtain
\[ \Delta^- h(y)\leq \Delta^- {\mathfrak V}(y)\leq\Delta^+
{\mathfrak V}(y)\leq\Delta^+ h(y),\] where $\Delta^{+}$ and
$\Delta^-$ denote  right- and left-derivatives, respectively. Now
\begin{eqnarray*}
\Delta^+ h(y)&=& \lim_{\eps\to 0} \frac{h(y+\eps)-h(y)}{\eps} \ =\
\lim_{\eps\to 0} \frac{1}{\eps}\, \eit{
V(t,(y+\eps)\yqy_t)-V(t,y\yqy_t)}+ \scl{\hqy}{\EN_T}\\&\leq&
\liminf_{\eps\to 0} \, \left( -\frac{1}{\eps} \right) \, \eit{\eps \yqy_t
I(t,(y+\eps)\yqy_t)}+\scl{\hqy}{\EN_T}\\ &=& -\eit{\yqy_t
I(t,y\yqy_t)}+\scl{\hqy}{\EN_T}
\end{eqnarray*} by the Monotone Convergence Theorem. Similarly, we get
\[ \Delta^- h(y)\ \geq\  \limsup_{\eps\to 0 } \EE\left[ -\int_0^T \yqy_t
I(t,(y-\eps)\yqy_t)\,\mu(dt)\right]+<\hqy,\EN_t>.\] Let $y_0$  be
the constant from $\Gamma_4$, Lemma \ref{ae1}, corresponding to
some $\mathrm{AE}[U]\leq\gamma<1$ a.s. Then
\begin{eqnarray*}| \yqy_t I(t,(y-\eps)\yqy_t)|&\leq& |\yqy_t
I(t,(y-\eps)\yqy_t)| \,\inds{\yqy_t\leq y_0/y}\  +\  |\yqy_t
I(t,(y-\eps)\yqy_t)| \,\inds{\yqy_t> y_0/y}.\end{eqnarray*} We fix
$\eps_0$ and observe that for $\eps<\eps_0$, by Lemma \ref{ae1},
the second part is dominated by \begin{equation}\lprint{star}
\frac{1}{y-\eps_0}\frac{\gamma}{1-\gamma} V(t,(y-\eps_0)\yqy_t) \
\leq\  \frac{1}{y-\eps_0}\frac{\gamma}{1-\gamma}
CV(t,y\yqy_t),\end{equation} for some constant $C$. This last
expression is in $L^1(\mu\otimes\PP)$, by finiteness of
${\mathfrak V}(\cdot)$. On the other hand, the first part in
(\ref{star}) is dominated by $K_1(\frac{y-\eps_0}{y}y_m)\yqy_t$,
which is in $L^1(\mu\otimes\PP)$ by the supermartingale property
of $\yqy$. Having prepared the ground for the Dominated
Convergence Theorem, we can let $\eps\to 0$ and obtain
\[ \Delta^- h(y)\ \geq\  <\hqy,\EN_T>-\eit{\yqy_t I(t,y\yqy_t)},\]
completing the proof of the proposition.
\end{proof}

\begin{lem} The dual value function ${\mathfrak V}(\cdot)$ has the following asymptotic
behavior: \begin{enumerate}
\item[(i)] ${\mathfrak V}'(0+)=-\infty$,
\item[(ii)] ${\mathfrak V}'(\infty)\in[\inf_{\QQ\in\DD} <\QQ,\EN_T>,
\sup_{\QQ\in\DD} <\QQ,\EN_T>]$.
\end{enumerate}
\end{lem}
\begin{proof}\nopagebreak \ \nopagebreak
\begin{enumerate}
\item[(i)]  Suppose first there is a minorant $\uv(\cdot)$ of $V(\cdot,\cdot)$ such that
$\uv(0+)=\infty$.  Letting $y\to 0$ in (\ref{Jens}),  we get
${\mathfrak V}(0+)=\infty$ and, by convexity, ${\mathfrak
V}'(0+)=-\infty$.

In the case when $\uv(0+)<\infty$ for each minorant $\uv(\cdot)$
of $V(\cdot,\cdot)$, we can easily construct a majorant
$\ov(\cdot)$ such that  $\ov(0+)<\infty$, using the properties of
finctions $K_1$ and $K_2$ from Definition \ref{DefUtility}. We
pick such a majorant $\ov(\cdot)$, a minorant $\uv(\cdot)$, set
$\oi(\cdot)=-\ov'(\cdot)$, $D=\ov(0+)-\uv(0+)$, and choose
$\QQ\in\DD$. Then, with $\rho=\norm{\EN_T}_{\linf}$,
\begin{eqnarray*}
-{\mathfrak V}'(y)&\geq& \frac{{\mathfrak V}(0+)-{\mathfrak
V}(y)}{y}\ \geq\ \frac{1}{y}\left[
(\uv(0+)-\ov(0+))+\ov(0+)-{\mathfrak V}(y)\right]\\ &\geq&
\frac{-D-\rho y}{y}+\frac{\ov(0+)-\eit{\ov(y\yq_t)}}{y}\\&\geq&
\frac{-D-\rho y}{y}+\eit{\yq_t\oi(y\yq_t)}\longrightarrow\infty,\
\ \text{as $y\to\infty$,}
\end{eqnarray*}
by the Monotone convergence theorem.
\item[(ii)]
By l'H\^ opital's rule we have
\begin{eqnarray*}
{\mathfrak V}'(\infty)  &=&  \lim_{y\to\infty}\frac{{\mathfrak
V}(y)}{y}\ =\ \lim_{y\to\infty}\frac{\inf_{\QQ\in\DD} \left(
\eit{V(t,y\yq_t)}+ y<\QQ, \EN_T>\right) }{y}\\ & \in & \left[L +
\inf_{\QQ\in\DD}<\QQ, \EN_T>,\  L+ \sup_{\QQ\in\DD}<\QQ,
\EN_T>\right],
\end{eqnarray*}
where  $L\triangleq\lim_{y\to\infty}\frac{1}{y}\inf_{\QQ\in\DD}
\eit{V(t,y\yq_t)}$. From the Definition \ref{DefUtility} of the
utility function
 we
read $\partial_2 V(t,y)\leq -(K_1)^{-1}(y)\to 0$ when
$y\to\infty$, so for an $\eps>0$ we can find a constant $C(\eps)$
such that  $-V(t,y)\leq C(\eps)+\eps y$ for all $t\in[0,T]$ and
all $y>0$. To finish the proof, we denote by ${\mathfrak
V}_0(\cdot)$ the (strictly convex, decreasing) value function of
the dual optimization problem (\ref{DP}) when $\EN_T\equiv 0$.
Then the decrease of ${\mathfrak V}_0(\cdot)$ and L'H\^ opital's
rule imply
\begin{eqnarray*}
0&\leq&-{\mathfrak V}_0'(\infty)=\lim_{y\to\infty}
\frac{-{\mathfrak V}_0(y)}{y}=\lim_{y\to\infty} \sup_{\QQ\in\DD}
\frac{1}{y}\eit{-V(t,y\yq_t)}=-L\\&\leq&
\lim_{y\to\infty}\sup_{\QQ\in\DD}\frac{1}{y}\eit{(C(\eps)+\eps
y\yq_t)}\ \leq\ \lim_{y\to\infty} \EE\int_0^T
(\frac{C(\eps)}{y}+\eps)\,\mu(dt)=\eps .
\end{eqnarray*}
Consequently, $L=0$, and the claim follows.
\end{enumerate}\vskip-0.6cm
\end{proof}

\subsection{Proof of the Main Theorem \ref{MainTheorem}}
In this subsection we combine the preceding lemmas and
propositions, to complete the proof of Theorem \ref{MainTheorem}.
 \begin{itemize}
\item[(i)] By the concavity of $U(t,\cdot)$ and the Standing Assumption
\ref{StandingAssumption}, we   deduce that ${\mathfrak
U}(x)<\infty$ for any $x>0$. For $x>0$ we define $c(t)\triangleq
x,\ \forall\,t\in [0,T]$. Then $c\in\AA^{\mu}(x+\EN)$, because the
constant consumption-rate process $c(\cdot)\equiv x$  can be
financed by the trivial portfolio $H\equiv 0$ and initial wealth
only. Since
\begin{eqnarray*} \hskip0.9cm{\mathfrak U}(x) \,\geq\, \eit{U(t,c(t))}=\eit{U(t,x)}\,
\geq\, \eit{\uu(x)}=\uu(x)\,>\,-\infty,\end{eqnarray*} we conclude
that $|{\mathfrak U}(x)|<\infty$ for all $x>0$. The assertion that
$|{\mathfrak V}(y)|<\infty$ for all $y>0$ is the content of Lemma
\ref{finv}.
\item[(ii)] ${\mathfrak V}(\cdot)$ is continuously differentiable by Proposition
\ref{vprime}. From the conjugacy relation in Lemma \ref{conj} and
the properties of convex conjugation (see Theorem 26.5 in
\cite{Roc70}), we deduce the continuous differentiability of
${\mathfrak U}(\cdot)$.
\item[(iii)] Follows from Lemma
\ref{conj} and the properties of convex conjugation (see Theorem
12.2. in \cite{Roc70}).
\item[(iv)] The assertion is a direct consequence of Lemma
\ref{vprime} and the properties of convex conjugation (see Theorem
26.5. in \cite{Roc70} ).
\item[(vi)] Follows from Lemma \ref{vprime}.
\item[(v)] The dual problem has an essentially unique solution
$\hq^y\in\DD$ for any $y>0$, by Proposition \ref{exsol} and Remark
\ref{finvy}. To establish the result for the primal problem, we
pick $x>0$, a solution $\hq^y$ of the dual problem corresponding
to $y={\mathfrak U}'(x)$ and define $\hc^x(t)\triangleq
I(t,y\yqy_t),$ for all $t\in [0,T]$. Then the relation
$-{\mathfrak V}'(y)=({\mathfrak U}'(\cdot))^{-1}(y)$, $y>0$ (see
\cite{Roc70}, Theorem 26.6) and  Proposition \ref{vprime} give \[
\eit{\hc^x(t)\yqy_t}=-{\mathfrak
V}'(y)+\scl{\hq^y}{\EN_T}=x+\scl{\hq^y}{\EN_T},\] so for any
$\QQ\in\DD$, by Proposition \ref{difv}, \ea{
\eit{\hc^x(t)\yq_t}&\leq&\eit{\hc^x(t)\yqy_t} +
\scl{\QQ}{\EN_T}-\scl{\hq^y}{\EN_T} = x+\scl{\QQ}{\EN_T}.} Thus
$\hc^x(\cdot)\in\AA(x+\EN)$ by the characterization of admissible
consumption processes in Proposition \ref{OU}.

Having established the admissibility of $\hc^x(\cdot)$, we note
that \ea{ \eit{U(t,\hc^x(t))}&=&\eit{V(t,y\yqy_t)} + \eit{y\yqy_t
I(t,y\yqy_t)} \\ &=& {\mathfrak V}(y)-y{\mathfrak
V}'(y)={\mathfrak U}(x),} by the conjugacy relation (iii), the
expression of the derivative of the dual value function (v), and
the definition of $y$. This closes the duality gap and proves the
optimality of $\hc^x(\cdot)$.
\end{itemize}

%------------------------------------------
% Part VII - Bibliography
%------------------------------------------

\def\cprime{$'$} \def\cprime{$'$}
\providecommand{\bysame}{\leavevmode\hbox
to3em{\hrulefill}\thinspace}

% The end
\end{document}